\documentclass[runningheads]{llncs}

\usepackage{amsfonts,proof,qtree,amsmath,amssymb,frameit,mathtools,algpseudocode,algorithm}
\usepackage{latexsym}
\usepackage{graphicx}
\usepackage[usenames,dvipsnames]{color}
\usepackage{listings}
\usepackage{float}
\usepackage{multirow}
\usepackage{booktabs}
\usepackage{mathrsfs}
\usepackage{mathpartir}
\usepackage{dsfont}
\usepackage{stmaryrd}
\usepackage{url}
\usepackage{textcomp}
\usepackage{hyperref}
\usepackage{alltt}
\usepackage{bbm}
\usepackage{alltt}
\usepackage{xspace}
\usepackage{enumitem}
\usepackage{lipsum}
\usepackage{wrapfig}
\usepackage[usenames,dvipsnames]{xcolor}

\usepackage{cite}
\usepackage{comment}
\usepackage{ebproof}
\usepackage{subfig}
\usepackage{graphics}
\usepackage{tikz}
\usetikzlibrary{arrows,decorations.markings}
\usepackage{enumitem}
\usepackage{hhline}
\usepackage{colortbl}
\usepackage{ebproof}
\usepackage{centernot}

\renewcommand{\figurename}{Fig.}

\hypersetup{colorlinks,
  linkcolor=ACMDarkBlue,
  citecolor=ACMPurple,
  urlcolor=ACMDarkBlue,
  filecolor=ACMDarkBlue}

\let\subparagraph\paragraph

\usepackage{titlesec}
\titlespacing*{\section}{0pt}{*1.5}{*1}
\titlespacing*{\subsection}{0pt}{*1.5}{*1}
\titlespacing*{\paragraph}{0pt}{*0.5}{*0.5}
\setlist[itemize]{leftmargin=*}
\setlist[enumerate]{leftmargin=*}

\makeatletter 
\def\arcr{\@arraycr}
\makeatother

\definecolor[named]{ACMBlue}{cmyk}{1,0.1,0,0.1}
\definecolor[named]{ACMYellow}{cmyk}{0,0.16,1,0}
\definecolor[named]{ACMOrange}{cmyk}{0,0.42,1,0.01}
\definecolor[named]{ACMRed}{cmyk}{0,0.90,0.86,0}
\definecolor[named]{ACMLightBlue}{cmyk}{0.49,0.01,0,0}
\definecolor[named]{ACMGreen}{cmyk}{0.20,0,1,0.19}
\definecolor[named]{ACMPurple}{cmyk}{0.55,1,0,0.15}
\definecolor[named]{ACMDarkBlue}{cmyk}{1,0.58,0,0.21}
\definecolor[named]{PaleGreen}{RGB}{196, 255, 231}
\definecolor[named]{PaleOrange}{RGB}{255, 213, 169}
\definecolor{intnull}{RGB}{213,229,255}

\definecolor{shadecolor}{gray}{1.00}
\definecolor{ddarkgray}{gray}{0.5}
\definecolor{darkgray}{gray}{0.30}
\definecolor{light-gray}{gray}{0.91}

\newcommand{\graybox}[1]{\colorbox{light-gray}{#1}}

\newcommand{\gbm}[1]{\graybox{${#1}$}}

\newcommand{\angled}[1]{\left\langle{#1}\right\rangle}
\newcommand{\paren}[1]{\left({#1}\right)}

\newcommand{\ie}{\emph{i.e.}\xspace}

\newcommand{\eg}{\emph{e.g.}\xspace}

\newcommand{\etal}{\emph{et~al.}\xspace}

\newcommand{\dom}[1]{\mathsf{dom}\paren{#1}\xspace}

\newcommand{\aka}{\textit{a.k.a.}\xspace}

\newcommand{\wrt}{\emph{wrt.}\xspace}
\newcommand{\Iff}{\emph{iff}\xspace}

\newcommand{\eqdef}{\triangleq}

\newcommand{\cached}[1]{#1}

\newcommand{\denot}[1]{\llbracket{#1}\rrbracket}

\newcommand{\ifext}[2]{\ifdefined\extflag{#1}\else{#2}\fi}

\newcommand{\tname}[1]{\textsc{#1}\xspace}

\newcommand{\tool}{\tname{SuSLik}}
\newcommand{\toolx}{\tname{ROBo\tool}}
\newcommand{\logic}{SSL\xspace}
\newcommand{\logicx}{BoSSL\xspace}

\newcommand{\logicfullx}{Borrowing Synthetic Separation Logic\xspace}

\definecolor{pblue}{rgb}{0.13,0.13,1}
\definecolor{pgreen}{rgb}{0,0.5,0}
\definecolor{pred}{rgb}{0.9,0,0}
\definecolor{pgrey}{rgb}{0.46,0.45,0.48}

\definecolor{ckeyword}{HTML}{7F0055}
\definecolor{ccomment}{HTML}{3F7F5F}
\definecolor{cnumber}{HTML}{2A0099}

\lstdefinelanguage{SynLang}{
  keywords={new, let, if, else, null, return, while},
  ndkeywords={bool, int, void, loc},
  mathescape=true,
  showspaces=false,
  showtabs=false,
  breaklines=true,
  showstringspaces=false,
  breakatwhitespace=true,
  lineskip=-0.9pt,
  morecomment=[l]{//}, 
  morecomment=[s]{/*}{*/}, 
  basewidth={0.54em, 0.4em},%
  basicstyle=\footnotesize\ttfamily,
  keywordstyle={\color{ACMPurple}\ttfamily\bfseries},
  ndkeywordstyle={\color{pblue}\ttfamily\bfseries},
  commentstyle={\color{ccomment}\itshape},
  numbers=none,
}

\lstdefinestyle{numbers}
{
  numbers=left,
  numberstyle=\scriptsize\sf,
  xleftmargin=15pt
}

\newcommand{\code}[1]{\lstinline[language=SynLang,basicstyle=\small\ttfamily,mathescape=true]{#1}}

\newcommand{\set}[1]{\left\{{#1}\right\}}
\newcommand{\many}[1]{\overline{#1}}

\newcommand{\head}[1]{\emph{#1}}

\makeatletter
\protected\def\ccell#1#{%
  \kern-\fboxsep
  \@ccell{#1}%
}
\def\@ccell#1#2#3{%
  \colorbox#1{#2}{#3}%
  \kern-\fboxsep
}
\makeatother

\def\mut{\mathsf{M}} 
\def\imm{\mathsf{RO}} 
\def\ann{a}      
\def\iiann{\alpha} 
\newcommand{\bimm}[1]{{\mathsf{#1}}}
\def\iset{Perm}

\def\seq{=}           

\newcommand{\iann}[1]{\mcode{ {#1}}} 
\newcommand{\writeontop}[2]{\mathrel{\stackrel{\makebox[0pt]{\mbox{\tiny{#2}}}}{#1}}}
\newcommand{\spointsto}[2]{{#1}~{\pts}~{#2}}

\newcommand{\ispointsto}[3]{#1{\writeontop{\ \!\pts\ \!}{\iann{#3}}}{#2}} 
\newcommand{\ipointsto}[3]{\ispointsto{\angled{#1}}{#2}{#3}} 
\newcommand{\iipred}[3]{\pred{#1}(#2, #3)}
\newcommand{\ipred}[3]{\pred{#1}(#2)} 
\newcommand{\spred}[2]{\pred{#1}(#2)}
\newcommand{\iblock}[3]{\left[{#1}, {#2}\right]\!^{#3}}

\def\implies{\Rightarrow}

\def\evalimm{\imm}
\def\evalmut{\mathsf{Mut}}

\def\better{\cellcolor{PaleGreen}}
\def\worse{\cellcolor{PaleOrange}}

\def\stronger{\better YES}
\def\same{same}

\newcommand{\oprulen}[1]{[\underline{{\bf \scriptstyle \rulename{OP}-#1}}]}

\newcommand{\ebpsem}[2]{
	\begin{prooftree}
		\Hypo{\oprulen{#1}}
		\Infer[no rule]1{{\begin{prooftree}[rule style=simple] #2 \end{prooftree}}}
	\end{prooftree}
}

\def\ilset{R}

\newcommand{\isatsl}[3]{{#1}{\vDash}^{\ctx,{#3}}_{\interp} {#2}}

\newcommand{\pts}{\mapsto}
\newcommand{\True}{\mathsf{true}}

\newcommand{\bl}[1]{{\color{ACMDarkBlue}{#1}}}
\newcommand{\asn}[1]{{\bl{\set{#1}}}}

\newcommand{\pred}[1]{\mathsf{#1}}

\newcommand{\trans}[3]{\left.\bl{#1} \!\leadsto\! \bl{#2} \right| #3}

\newcommand{\teng}[2]{\bl{#1}\leadsto \bl{#2}}

\newcommand{\ieqTags}{\seq}

\newcommand{\subst}{\sigma}
\newcommand{\subs}[2]{#2/#1}

\newcommand{\dotcup}{\ensuremath{\mathaccent\cdot\cup}}
\newcommand{\union}{\mathbin{\dotcup}}
\newcommand{\seteq}{=}

\newcommand{\slf}[1]{\mathcal{#1}}
\newcommand{\slP}{\slf{P}}
\newcommand{\slQ}{\slf{Q}}
\newcommand{\interp}{\slf{I}}

\newcommand{\ls}{\pred{ls}}
\newcommand{\lseg}{\pred{lseg}}

\newcommand{\nxt}{\mathit{nxt}}
\newcommand{\sep}{\ast}
\newcommand{\Sep}{\text{\scalebox{1.8}{$\sep$}}}
\newcommand{\emp}{\mathsf{emp}}
\newcommand{\block}[2]{\left[{#1}, {#2}\right]}

\newcommand{\rulename}[1]{\textsc{#1}}
\newcommand{\framer}{\rulename{Frame}\xspace}
\newcommand{\hunir}{\rulename{UnifyHeaps}\xspace}

\newcommand{\empr}{\rulename{Emp}\xspace}

\newcommand{\writer}{\rulename{Write}\xspace}
\newcommand{\writeri}{\rulename{WriteRO}\xspace}
\newcommand{\sleft}{\rulename{SubstLeft}\xspace}

\newcommand{\allocr}{\rulename{Alloc}\xspace}
\newcommand{\freer}{\rulename{Free}\xspace}
\newcommand{\applyr}{\rulename{Call}\xspace}

\newcommand{\invkind}{\rulename{Open}\xspace}

\newcommand{\ev}[1]{\mathsf{EV}\left({#1}\right)}

\newcommand{\vars}[1]{\mathsf{Vars}\left({#1}\right)}
\newcommand{\gv}[1]{\mathsf{GV}\left({#1}\right)}

\newcommand{\env}{\Gamma}
\newcommand{\ctx}{\Sigma}
\newcommand{\fctx}{\Delta}
\newcommand{\indp}{\mathcal{D}}
\newcommand{\fun}{\mathcal{F}}

\newcommand{\pclause}[3]{\angled{{#1}, \set{{#2}, {#3}}}}
\newcommand{\sel}{e}

\newcommand{\valid}[1]{\vdash #1}

\newcommand{\kw}[1]{{\tt\bf{\color{ACMPurple}{#1}}}~}

\newcommand{\Letz}{\kw{let}}
\newcommand{\Ifz}{\kw{if}}
\newcommand{\malloc}[1]{{\tt{malloc}}(#1)}
\newcommand{\free}[1]{{\tt{free}}(#1)}

\newcommand{\Elsez}{\kw{else}}
\newcommand{\deref}[1]{*{#1}}

\newcommand{\off}{\iota}
\newcommand{\size}{\mathit{n}}
\newcommand{\ccd}[1]{\text{\code{#1}}}

\newcommand{\mcode}[1]{{\ensuremath{\tt #1}}}

\newcommand{\prog}{c}

\newcommand{\cskip}{{\tt{skip}}}
\newcommand{\cerror}{{\tt{err}}}

\newcommand{\opsem}[3]{{#1}; {#2} \rightsquigarrow {#3}}
\newcommand{\Opsem}{\rightsquigarrow}
\newcommand{\Opsemt}{\rightsquigarrow^{*}}
\newcommand{\opsemt}[3]{{#1}; {#2} \rightsquigarrow^{*} {#3}}
\newcommand{\satsl}[2]{{#1}{\vDash}^{\ctx}_{\interp} {#2}}

\newcommand{\hide}[1]{}
\newcommand{\bientails}{\dashv\vdash}

\begin{document}

\title{
Concise Read-Only Specifications for 
\\
Better Synthesis of Programs with Pointers
\ifext{}
{\vspace{-5pt}}
}

\ifext{
\subtitle{Extended Version
\vspace{-15pt}
}
}
{}

\ifext{
\author{
  Andreea Costea\inst{1} \and
  Amy Zhu\inst{2} \and
  Nadia Polikarpova\inst{3} \and
  Ilya Sergey\inst{4,1}
\vspace{-5pt}
}
}
{
\author{
  Andreea Costea\inst{1} \and
  Amy Zhu\inst{2}\thanks{Work done during an 
  internship at NUS School of Computing in Summer 2019.} \and
  Nadia Polikarpova\inst{3} \and
  Ilya Sergey\inst{4,1}
\vspace{-5pt}
}
}

\institute{
  School of Computing, National University of Singapore, Singapore
  \and
  University of British Columbia, Vancouver, Canada
  \and
  University of California, San Diego, USA
  \and
  Yale-NUS College, Singapore
}

\maketitle

\vspace{-10pt}

\thispagestyle{plain}
\pagestyle{plain}

\begin{abstract}
%
%
In program synthesis there is a well-known trade-off between
\emph{concise} and \emph{strong} specifications: if a specification is
too verbose, it might be harder to write than the program; if it is
too weak, the synthesised program might not match the user's intent.
%
%
In this work we explore the use of annotations for restricting memory
access permissions in program synthesis, and show that they can make
specifications much stronger while remaining surprisingly concise.
%
%
Specifically, we enhance Synthetic Separation Logic (SSL), a framework
for synthesis of heap-manipulating programs, with the logical
mechanism of \emph{read-only borrows}.
%

\vspace{2pt}

We observe that this minimalistic and conservative SSL extension
benefits the synthesis in several ways, making it more
(a)~\emph{expressive} (stronger correctness guarantees are achieved with a
modest annotation overhead), (b)~\emph{effective} (it
produces more concise and easier-to-read programs),
(c)~\emph{efficient} (faster synthesis), and (d)~\emph{robust}
(synthesis efficiency is less affected by the choice of the search
heuristic).
We explain the intuition and provide formal treatment for read-only
borrows. We substantiate the claims (a)--(d) by describing our
quantitative evaluation of the borrowing-aware synthesis
implementation on a series of standard benchmark specifications for
various heap-manipulating programs.





\end{abstract}

\vspace{-15pt}

\section{Introduction}
\label{sec:intro}


Deductive program synthesis is a prominent approach to the generation of
correct-by-construction programs from their declarative
specifications~\cite{Delaware-al:POPL15,Kneuss-al:OOPSLA13,Polikarpova-al:PLDI16,Manna-Waldinger:TOPLAS80}.
%
With this methodology, one can represent searching for a program
satisfying the user-provided constraints as a proof search in a
certain logic.
Following this idea, it has been recently observed
\cite{Polikarpova-Sergey:POPL19}
that the synthesis of
correct-by-construction \emph{imperative heap-manipulating} programs
(in a language similar to C) can be implemented as a proof
search in a version of Separation Logic (SL)---a program logic
designed for modular verification of programs with
pointers~\cite{OHearn-al:CSL01,Reynolds:LICS02}.

SL-based deductive program synthesis based on \emph{Synthetic}
Separation Logic (\logic)~\cite{Polikarpova-Sergey:POPL19} requires
the programmer to provide a Hoare-style specification for a program of
interest. For instance, given the predicate
$\bl{\mcode{\ls(x, S)}}$, which denotes a symbolic heap corresponding
to a linked list starting at a pointer $\mcode{x}$, ending with
\code{null}, and containing elements from the set $\mcode{S}$,
one can specify the behaviour of the procedure for copying a linked
list as follows:
%
%
\begin{equation}
  \label{eq:listcopy}
\mcode{
  \asn{r \pts x \sep \ls(x, S)}
  ~\text{\code{listcopy(r)}}~
  \asn{r \pts y \sep \ls(x, S) \sep \ls(y, S)}
}
\end{equation}

\hspace{-12pt}
\begin{center}
  \includegraphics[width=0.7\textwidth]{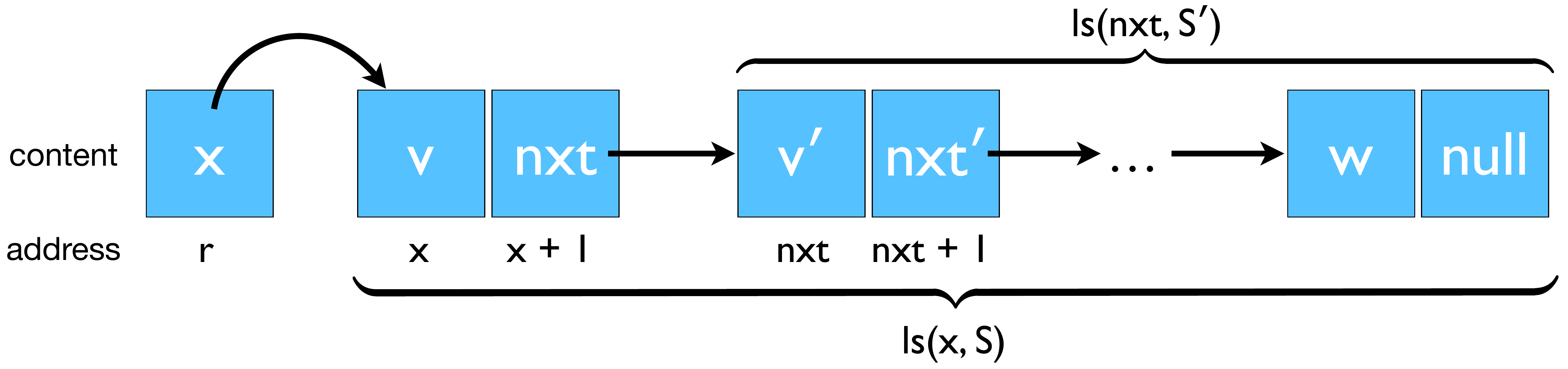}
\end{center}

The precondition of specification~\eqref{eq:listcopy}, defining the
shape of the initial heap, is illustrated by the figure above. It
requires the heap to contain a pointer $\bl{\mcode{r}}$, which is taken by the
procedure as an argument and whose stored value, $\bl{\mcode{x}}$, is the head
pointer of the list to be copied.
The list itself is described by the symbolic heap predicate instance
$\bl{\mcode{\ls(x, S)}}$, whose footprint is assumed to be \emph{disjoint}
from the entry $\bl{\mcode{r \pts x}}$, following the standard semantics of
the \emph{separating conjunction} operator
($\bl{\sep}$)~\cite{OHearn-al:CSL01}.
The postcondition asserts that the final heap, in addition to containing
the original list $\bl{\mcode{\ls(x, S)}}$, will contain a new list
starting from $\bl{\mcode{y}}$ whose contents $\bl{\mcode{S}}$ are the
same as of the original list, and also that the pointer $\bl{\mcode{r}}$
will now point to the head $\bl{\mcode{y}}$ of the list copy.
Our specification is incomplete: it allows, for example, duplicating
or rearranging elements. One hopes that such a program is unlikely to
be synthesised. In synthesis, it is common to provide incomplete
specs: writing complete ones can be as hard as writing the program
itself.


\vspace{-5pt}

\subsection{Correct Programs that Do Strange Things}
\label{sec:weak}

\begin{wrapfigure}[23]{r}{0.45\textwidth}
\vspace{-30pt}
\setlength{\abovecaptionskip}{-0pt}
\begin{center}
\begin{lstlisting}[style=numbers,xleftmargin=20pt,basicstyle=\footnotesize\ttfamily]
void listcopy (loc r) {
  let x = *r;
  if (x == 0) {
  } else {
    let v = *x;
    let nxt = *(x + 1);
    *r = nxt;
    listcopy(r);
    let y1 = *r;
    let y = malloc(2);
    *(x + 1) = y1;
    *r = y;
    *(y + 1) = nxt;
    *y = v;
  } }
\end{lstlisting}
\vspace{2pt}
\!\!\!\!
\includegraphics[width=0.45\textwidth]{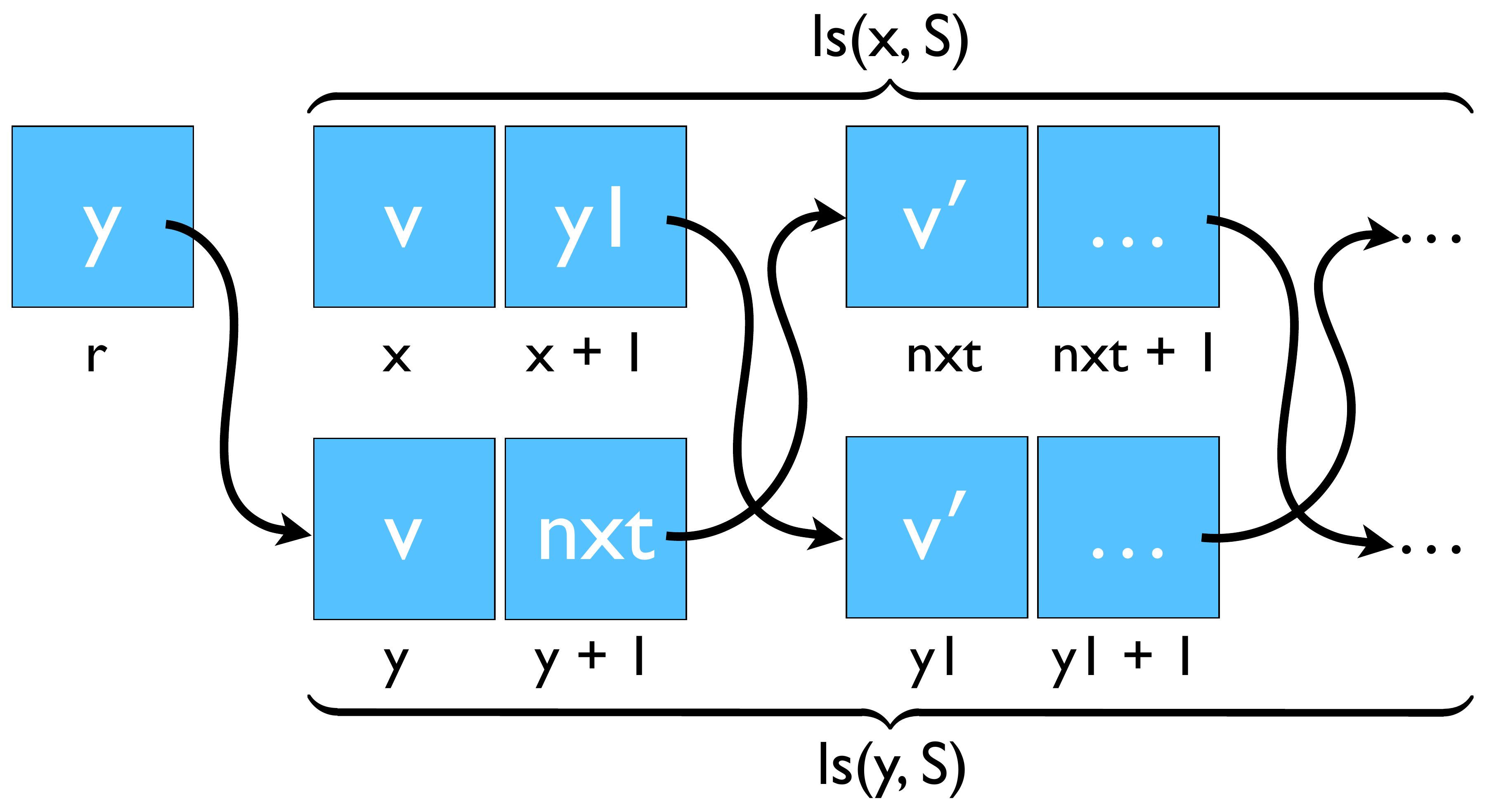}
\end{center}
\caption{Result program for spec~\eqref{eq:listcopy} and the shape of
  its final heap.}
\label{fig:copy}
\end{wrapfigure}
Provided the definition of the heap predicate $\bl{\ls}$ and the
specification~\eqref{eq:listcopy}, the \tool tool, an implementation
of the \logic-based synthesis~\cite{Polikarpova-Sergey:POPL19}, will
produce the program depicted in~\autoref{fig:copy}.
It is easy to check that this program satisfies the ascribed
spec~\eqref{eq:listcopy}. Moreover, it correctly duplicates the
original list, faithfully preserving its contents and the ordering.
However, an astute reader might notice a certain oddity in the way it
treats the initial list provided for copying.
According to the postcondition of~\eqref{eq:listcopy}, the value of
the pointer \code{r} stored in a local immutable variable \code{y1}
on line {\small\textsf{9}} is the head of the copy of the original
list's tail.
Quite unexpectedly, the pointer \code{y1} becomes the tail of the
original list on line {\small\textsf{11}}, while the \emph{original
  list's} tail pointer \code{nxt}, once assigned to \code{*(y + 1)}
on line {\small\textsf{13}}, becomes the tail of the \emph{copy}!
%


Indeed, the exercise in tail swapping is totally pointless: not only
does it produces less ``natural'' and readable code, but the resulting
program's locality properties are unsatisfactory; for instance, this
program cannot be plugged into a concurrent setting where multiple
threads rely on $\bl{\mcode{\ls(x, S)}}$ to be unchanged. 


The issue with the result in \autoref{fig:copy} is caused by
specification~\eqref{eq:listcopy} being \emph{too permissive}: it does not
prevent the synthesised program from \emph{modifying} the structure of
the initial list, while creating its copy. Luckily, the Separation
Logic community has devised a number of SL extensions that allow one
to impose such restrictions, like declaring a part of the provided
symbolic heap as
\emph{read-only}~\cite{Chargueraud-Pottier:ESOP17,Costea-al:PEPM14,HeuleLMS13:VMCAI13,Jacobs-al:NFM11,Boyland:SAS03,Dockins-al:APLAS09,Balabonski-al:TOPLAS16},
\ie, forbidden to modify by the specified code.
%
%

\subsection{Towards Simple Read-Only Specifications for Synthesis}
\label{sec:towards}

The main challenge of introducing read-only annotations (commonly also
referred to as \emph{permissions})\footnote{We will be using the words
  ``annotation'' and ``permission'' interchangeably.} into Separation
Logic lies in establishing the discipline for performing sound
accounting in the presence of mixed read-only and mutating heap
accesses by different components of a program.

As an example, consider a simple symbolic heap
$\mcode{\asn{\ispointsto{x}{f}{\mut} \sep \ispointsto{r}{h}{\mut}}}$
that declares two \emph{mutable} (\ie, allowed to be written to)
pointers $\mcode{x}$ and $\mcode{r}$, that point to unspecified values
$\mcode{f}$ and $\mcode{h}$, correspondingly. With this symbolic heap,
is it safe to call the following function that modifies the
contents of $\mcode{r}$ but not of $\mcode{x}$?
%
%
\begin{equation}
  \label{eq:readx}
\mcode{
  \asn{\ispointsto{x}{f}{\imm} \sep \ispointsto{r}{h}{\mut}}
  ~\text{\code{readX(x, r)}}~
  \asn{\ispointsto{x}{f}{\imm} \sep \ispointsto{r}{f}{\mut}}
}
\end{equation}
The precondition of \code{readX} requires a weaker form of access
permission for~$\mcode{x}$ (read-only, $\bl{\imm}$), while the considered heap
asserts a stronger \emph{write} permission ($\bl{\mut}$). It
should be possible to satisfy \code{readX}'s requirement by providing
the necessary read-only permission for $\mcode{x}$.
In order to do so, we will need to agree on a discipline to ``adapt''
the caller's \emph{write}-permission $\bl{\mut}$ to the callee's
\emph{read-only} permission $\bl{\imm}$.
While seemingly trivial, if implemented na\"{i}vely, accounting of
read-only permissions in Separation Logic might undesirably compromise
either soundness or completeness of the logical reasoning.





A number of proposals for logically sound interplay between write- and
read-only access permissions in the presence of function calls has
been described in the
literature~\cite{HeuleLMS13:VMCAI13,Muller-al:VMCAI16,David-Chin:OOPSLA11,Costea-al:PEPM14,
  Chargueraud-Pottier:ESOP17,Boyland:SAS03,Bornat-al:POPL05}. Some of
these works manage to maintain the simplicity of having only
\emph{mutable}/\emph{read-only} annotations when confined to the
sequential
setting~\cite{Chargueraud-Pottier:ESOP17,David-Chin:OOPSLA11,Costea-al:PEPM14}.
More general (but harder to implement) approaches rely on
\emph{fractional permissions}~\cite{Boyland:SAS03,Le-Hobor:ESOP18}, an
expressive mechanism for permission accounting, with primary
applications in concurrent
reasoning~\cite{Bornat-al:POPL05,Leino-al:FOSAD09}.
We started this project by attempting to adapt some of those
logics~\cite{David-Chin:OOPSLA11,Costea-al:PEPM14,
  Chargueraud-Pottier:ESOP17} as an extension of \logic in order to reap the
benefits of read-only annotations for the synthesis of sequential program.
%
%
%
The main obstacle we encountered involved definitions of
inductive heap predicates with \emph{mixed} permissions. For instance,
how can one specify a program that modifies the contents of a linked
list, but not its structure?
Even though it seemed possible to enable this treatment of predicates
via permission multiplication~\cite{Le-Hobor:ESOP18}, developing
support for this machinery on top of existing \tool infrastructure was
a daunting task.
Therefore, we had to look for a technically simpler solution.


\subsection{Our Contributions}
\label{sec:contribs}

\paragraph{Theoretical Contributions.~}

Our main conceptual innovation is the idea of instrumenting \logic
with symbolic \emph{read-only borrows} to enable faster and more
predictable program synthesis. Borrows are used to annotate symbolic
heaps in specifications, similarly to abstract fractional permissions
from the deductive verification tools, such as \tname{Chalice} and
\tname{VeriFast}~\cite{LeinoM09,Jacobs-al:NFM11,HeuleLMS13:VMCAI13}.
They enable simple but principled lightweight threading of heap access
permissions from the callers to callees and back, while enforcing
\emph{read-only} access whenever it is required.
For basic intuition on read-only borrows, consider the specification
below:
%
%
\begin{equation}
  \label{eq:readxy}
\mcode{
  \asn{\ispointsto{x}{f}{\bimm{a}} \sep \ispointsto{y}{g}{\bimm{b}}
    \sep \ispointsto{r}{h}{\mut}}
  ~\text{\code{readXY(x, y, r)}}~
  \asn{\ispointsto{x}{f}{\bimm{a}} \sep \ispointsto{y}{g}{\bimm{b}}
    \sep \ispointsto{r}{(f + g)}{\mut}}
}
\end{equation}
%
%
The precondition requires a heap with three pointers, $\mcode{x}$,
$\mcode{y}$, and $\mcode{r}$, pointing to unspecified $\mcode{f}$,
$\mcode{g}$, and $\mcode{h}$, correspondingly. Both $\mcode{x}$ and
$\mcode{y}$ are going to be treated as read-only, but now, instead of
simply annotating them with $\bl{\imm}$, we add \emph{symbolic
  borrowing annotations} $\bl{\bimm{a}}$ and $\bl{\bimm{b}}$.
%
The semantics of these borrowing annotations is the same as that of other ghost variables
(such as $\mcode{f}$).
In particular, the \emph{callee} must behave correctly for any valuation of $\bl{\bimm{a}}$ and $\bl{\bimm{b}}$,
which leaves it no choice but to treat the corresponding heap fragments as read-only
(hence preventing the heap fragments from being written).
On the other hand, from the perspective of the \emph{caller},
they serve as formal parameters that are substituted with actuals of caller's choosing:
for instance,  when invoked with a caller's symbolic heap
$\asn{\mcode{\ispointsto{x}{1}{\mut} \sep \ispointsto{y}{2}{\bimm{c}} \sep
  \ispointsto{r}{0}{\mut}}}$ (where $\bl{\bimm{c}}$ denotes a read-only
borrow of the caller), \code{readXY} is guaranteed to ``restore'' the
same access permissions in the postcondition, as per the substitution
$[\mut/\bimm{a}, \bimm{c}/\bimm{b}]$.
%
%
The example above demonstrates that read-only borrows are
straightforward to compose when reasoning about code with function
calls.
%
%
They also make it possible to define \emph{borrow-polymorphic}
inductive heap predicates, \eg, enhancing $\bl{\ls}$ from
spec~\eqref{eq:listcopy} so it can be used in specifications with
mixed access permissions on their components.\footnote{We will present
  borrow-polymorphic inductive heap predicates in
  \autoref{sec:predicates}.}
Finally, read-only borrows make it almost trivial to adapt the
existing \logic-based synthesis to work with read-only access
permissions; they reduce the complex permission
\emph{accounting} to easy-to-implement permission \emph{substitution}.

\paragraph{Practical Contributions.~}

Our first practical contribution is \toolx---an enhancement of the
\tool synthesis tool~\cite{Polikarpova-Sergey:POPL19} with support for
read-only borrows, which required us to modify less than 100 lines of
the original code.



Our second practical contribution is the extensive evaluation of
synthesis with read-only permissions, on a standard benchmark suite of
specifications for heap-manipulating programs.
We compare the behaviour, performance, and the outcomes of the
synthesis when run with the standard (``all-mutable'') specifications
and their analogues instrumented with read-only permissions wherever
reasonable.
By doing so, we substantiate the following claims regarding the
practical impact of using read-only borrows in \logic specifications:

\begin{itemize}
\item First, we show that synthesis of read-only specifications is
  more \emph{efficient}: it does \emph{less backtracking} while
  searching for a program that satisfies the imposed constraints,
  entailing better performance.

\item Second, we demonstrate that borrowing-aware synthesis is more
  \emph{effective}: specifications with read-only annotations lead to
  more concise and human-readable programs, which do not perform
  redundant operations.

\item Third, we observe that read-only borrows increase
  \emph{expressivity} of the synthesis: in most of the cases enhanced
  specifications provide stronger correctness guarantees for the
  results, at almost no additional annotation overhead.

\item Finally, we show that read-only borrows make the synthesis more
  \emph{robust}: its results and performance are less likely to be
  affected by the unification order or the order of the attempted rule
  applications during the search.
\end{itemize}

\paragraph{Paper Outline.~}
We start by showcasing the intricacies and the virtues of \logic-based
synthesis with read-only specifications in \autoref{sec:overview}.
We provide the formal account of read-only borrows and present the
modified \logic rules, along with the soundness argument and its
implications in \autoref{sec:logic}.
%
%
We report on the implementation and evaluation of the enhanced
synthesis procedure in \autoref{sec:evaluation}.
We conclude with a discussion of the shortcomings of read-only borrows
in \autoref{sec:limitations} and a comparison to related work
in \autoref{sec:related}.



\section{Program Synthesis with Read-Only Borrows}
\label{sec:overview}

We introduce the enhancement of \logic with read-only borrows by
walking the reader through a series of small but characteristic
examples of deductive synthesis with separation logic. We provide the
necessary background on \logic in \autoref{sec:ssl}; the readers
familiar with the logic may want to skip to
\autoref{sec:reduce-nondet}.

\subsection{Basics of \logic-based Deductive Program Synthesis}
\label{sec:ssl}

In a deductive Separation Logic-based synthesis, a client provides a
specification of a function of interest as a pair of pre- and
post-conditions, such as
$\mcode{\asn{\mcode{\slP}}~\text{\code{void foo(loc x, int
    i)}}~\asn{\mcode{\slQ}}}$.
The precondition $\bl{\slP}$ constrains the symbolic state
necessary to run the function safely (\ie, without crashes), while
the postcondition $\bl{\slQ}$ constrains the resulting state at the end of
the function's execution.
A function body $\mcode{\prog}$ satisfying the provided specification is
obtained as a result of deriving the \logic statement, representing
the synthesis \emph{goal}:

\vspace{5pt}
\centerline{
  $\mcode{\set{\mcode{x}, \mcode{i}}; \trans{\asn{\slP}}{\asn{\slQ}}{\prog}}$
}
\vspace{5pt}

In the statement above, $\mcode{x}$ and $\mcode{i}$ are \emph{program
  variables}, and they are explicitly stated in the environment
$\env = \set{\mcode{x}, \mcode{i}}$. Variables that appear in
$\asn{\slP}$ and that are not program variables are called (logical)
\emph{ghost} variables, while the non-program variables that only
appear in $\asn{\slQ}$ are referred to as (logical) \emph{existential}
ones (\mcode{EV}). The meaning of the statement
$\env; \mcode{\trans{\asn{\slP}}{\asn{\slQ}}{\prog}}$ is the
\emph{validity} of the Hoare-style triple
$\mcode{\asn{\slP}~\prog~\asn{\slQ}}$
for all possible values of variables from
$\env$.\footnote{We often care only about the \emph{existence} of a
  program $\prog$ to be synthesised, not its specific shape. In those
  cases we will be using a shorter statement:
  $\env; \mcode{\teng{\asn{\slP}}{\asn{\slQ}}}$. }
Both pre- and postcondition contain a \emph{spatial} part describing
the shape of the symbolic state (spatial formulae are ranged over via
$\bl{\mcode{P}}$, $\bl{\mcode{Q}}$, and $\bl{\mcode{R}}$), and a
\emph{pure} part (ranged over via $\phi$, $\psi$, and $\xi$), which
states the relations between variables (both program and logical).
A derivation of an SSL statement is conducted by applying logical
rules, which reduce the initial goal to a trivial one, so it can be
solved by one of the \emph{terminal} rules, such as, \eg, the rule
\empr shown below:
%

\vspace{-10pt}
{\small{
\begin{mathpar}
\inferrule* [left=\empr]
{
\valid{\phi \Rightarrow \psi}
}
{
\env; \trans{\asn{\mcode{\phi; \emp}}}{\asn{\mcode{\psi; \emp}}}{\cskip}
}
\end{mathpar}
}}
\vspace{-15pt}

\noindent
That is, \empr requires that ($i$) symbolic heaps in both pre- and
post-conditions are empty and ($ii$) that the pure part $\phi$ of
the precondition implies the pure part $\psi$ of the
postcondition.
As the result, \empr ``emits'' a trivial program $\cskip$.
Some of the \logic rules are aimed at simplifying the goal, bringing
it to the shape that can be solved with \empr. For instance, consider the
following rules:

\vspace{-2pt}
{\small{
\hspace{-30pt}
\begin{tabular}{cc}
\begin{minipage}{0.5\linewidth}
\begin{mathpar}
\inferrule[\framer]
{
\ev{\env, \slP, \slQ} \cap \vars{\mcode{R}} = \emptyset
\\
\env; \trans{\asn{\mcode{\phi; P}}}{\asn{\mcode{\psi; Q}}}{\mcode{\prog}}
}
{
\env; \trans{\asn{\mcode{\phi; P \sep R}}}{\asn{\mcode{\psi; Q \sep R}}}{\mcode{\prog}}
}
\end{mathpar}
\end{minipage}
&
\begin{minipage}{0.5\linewidth}
\begin{mathpar}
\inferrule[\hunir]
{
\mcode{[\subst]R' = R}
\\
\emptyset \neq \dom{\subst} \subseteq \ev{\env, \slP, \slQ}
\\
\env; \trans{\asn{\mcode{\phi; P \sep R}}}{\mcode{[\subst]\asn{\psi; Q \sep R'}}}{\mcode{\prog}}
}
{
\env; \trans{\asn{\mcode{\phi; P \sep R}}}{\asn{\mcode{\psi; Q \sep R'}}}{\mcode{\prog}}
}
\end{mathpar}
\end{minipage}
\end{tabular}
}}
\vspace{10pt}

Neither of the rules \framer and \hunir ``adds'' to the program
$\mcode{\prog}$ being synthesised. However, \framer reduces the goal by
removing a matching part $\bl{\mcode{R}}$ (\aka \emph{frame}) from both the pre-
and the post-condition. \hunir non-deterministically picks a
substitution $\subst$, which replaces existential variables in a sub-heap
$\bl{\mcode{R'}}$ of the postcondition to match the corresponding
symbolic heap $\bl{\mcode{R}}$ in the precondition. Both of these rules make
choices with regard to what frame $\bl{\mcode{R}}$ to remove or which
substitution $\subst$ to adopt---a point that will be of importance
for the development described in \autoref{sec:reduce-nondet}.

Finally, the following (simplified) rule for producing a \emph{write}
command is \emph{operational}, as it emits a part of the program to be
synthesised, while also modifying the goal accordingly. The resulting
program will, thus, consist of the emitted store $\mcode{\deref{x} = e}$ of an
expression $\mcode{e}$ to the pointer variable $\mcode{x}$. The remainder is
synthesised by solving the sub-goal produced by applying the \writer rule.

\vspace{-10pt}
{\small{
\begin{mathpar}
\inferrule*[left=\writer]
{
\mcode{\vars{e} \subseteq \env}
\\
\mcode{e \neq e'}
\\
\env; \trans{{\asn{\mcode{\phi; {x} \pts e \sep P}}}}
         {\asn{\mcode{\psi; {x} \pts e \sep Q}}} {\mcode{\prog}}
}
{
\mcode{
\env; \trans{\asn{\phi; {x} \pts e' \sep P}}
{\asn{\psi; {x} \pts e \sep Q}}
{{\deref{x} = e; \prog}}
}
}
\end{mathpar}
}}
\vspace{-15pt}

As it is common with proof search, should no rule apply to an
intermediate goal obtained in one of the derivations,
the deductive synthesis will back-track, possibly discarding a
partially synthesised program fragment, trying alternative derivation
branches.
%
%
For instance, an application of \hunir that unifies wrong sub-heaps
might lead the search down a path to an unsatisfiable goal,
eventually making the synthesis back-track and leading to longer search.
Consider also a misguided application of \writer into a certain location,
which can cause the synthesizer to generate a less intuitive program
that ``makes up'' for the earlier spurious writes.
This is
precisely what we are going to fix by introducing read-only
annotations.

\subsection{Reducing Non-Determinism with Read-Only Annotations}
\label{sec:reduce-nondet}

Consider the following example adapted from the original \logic
paper~\cite{Polikarpova-Sergey:POPL19}.
While the example is intentionally artificial, it captures a frequent
synthesis scenario---non-determinism during synthesis.
This specification allows a certain degree of freedom in
how it can be satisfied:
\begin{equation}
  \label{eq:pick1}
  \mcode{\asn{ \ispointsto{x}{239}{} \sep \ispointsto{y}{30}{} }}
  ~\text{\code{void pick(loc x, loc y)}}~
  \mcode{\asn{z \leq 100; \ispointsto{x}{z}{} \sep \ispointsto{y}{z}{} }}
\end{equation}
It seems logical for the synthesis to start the program derivation by
applying the rule \hunir, thus reducing the initial goal to the one of
the form
\[
\mcode{
\set{\mcode{x}, \mcode{y}}; \teng{\asn{ \ispointsto{x}{239}{} \sep
    \ispointsto{y}{30}{}}} {\asn{239 \leq 100; \ispointsto{x}{239}{}
    \sep \ispointsto{y}{239}{} }}
}
\]
This new goal has been obtained by picking one particular substitution
$\subst = [239/\mcode{z}]$ (out of multiple possible ones), which
delivers two identical \emph{heaplets} of the form $\bl{\mcode{x} \pts
239}$ in pre- and postcondition.
It is time for the \writer rule to strike to fix the discrepancy
between the symbolic heap in the pre- and postcondition by emitting
the command \mcode{*y = 239} (at last, some executable code!), and
resulting in the following new goal (notice the change of $\bl{y}$-related
entry in the precondition):
\[
\mcode{
  \set{\mcode{x}, \mcode{y}};
  \teng{\asn{x \pts 239 \sep y \pts 239}} {\asn{239 \leq 100;
      x \pts 239 \sep y \pts 239}}
}
\]
What follows are two applications of the \framer rule to the common
symbolic heaps, leading to the goal:
$
  \set{\mcode{x}, \mcode{y}}
  \teng{\asn{\emp}}{\asn{239 \leq 100; \emp}}
$.
At this point, we are clearly in trouble.
The pure part of the precondition is simply $\bl{\True}$, while the
postcondition's pure part is $\bl{239 \leq 100}$, which is unsolvable.

Turns out that our initial pick of the substitution
$\subst = [239/\mcode{z}]$ was an unfortunate one, and we should
discard the series of rule applications that followed it, back-track
and adopt a different substitution, \eg, $\subst' = [30/\mcode{z}]$,
which will indeed result in solving our initial goal.\footnote{One
  might argue that it was possible to detect the unsolvable conjunct
  $239 \leq 100$ in the postcondition immediately after performing
  substitution, thus sparing the need to proceed with this derivation
  further. This is, indeed, a possibility, but in general it is hard
  to argue which of the heuristics in applying the rules will work
  better in general. We defer the quantitative argument on this matter
  until \autoref{sec:robust-unification}.}

Let us now consider the same specification for \mcode{pick} that has
been enhanced by explicitly annotating parts of the symbolic heap as
mutable and read-only:
\begin{equation}
  \label{eq:pick}
\mcode{
  \mcode{\asn{ \ispointsto{x}{239}{\mut} \sep \ispointsto{y}{30}{\imm} }}
  ~\text{\code{void pick(loc x, loc y)}}~
  \mcode{\asn{z \leq 100; \ispointsto{x}{z}{\mut} \sep
      \ispointsto{y}{z}{\imm} }}
}
\end{equation}
In this version of \logic, the effect of rules such as \empr, \framer,
and \hunir remains the same, while operational rules such as \writer,
become \emph{annotation-aware}.
Specifically, the rule \writer is now replaced by the following one:

%
{\small{
\begin{mathpar}
\inferrule*[left=\writeri]
{
\mcode{\vars{e} \subseteq \env}
\\
\mcode{e \neq e'}
\\
\env; \trans{\mcode{{\asn{\phi; \ispointsto{x}{e}{\mut} \sep P}}}}
         {\mcode{\asn{\psi; \ispointsto{x}{e}{\mut} \sep Q}}} {\mcode{\prog}}
}
{
\mcode{
\env; \trans{\asn{\phi; \ispointsto{x}{e'}{\mut} \sep P}}
{\asn{\psi; \ispointsto{x}{e}{\mut} \sep Q}}
{{\deref{x} = e; \prog}}
}
}
\end{mathpar}
}}
%
%


Notice how in the rule above the heaplets of the form
$\bl{\mcode{\ispointsto{x}{e}{\mut}}}$ are now annotated with the access
permission $\mut$, which explicitly indicates that the code may modify
the corresponding heap location.

Following with the example specification~\eqref{eq:pick}, we can
imagine a similar scenario when the rule \hunir picks the substitution
$\subst = [239/\mcode{z}]$.
Should this be the case, the next application of the rule \writeri
will not be possible, due to the \emph{read-only} annotation on the
heaplet $\bl{\mcode{\ispointsto{y}{239}{\imm}}}$ in the resulting sub-goal:
\begin{equation*}
\mcode{
\set{x, y}; \teng
    {\asn{\ispointsto{x}{239}{\mut} \sep \gbm{\ispointsto{\mcode{y}}{30}{\imm}} }}
    {\asn{z \leq 100; \ispointsto{x}{239}{\mut} \sep
        \gbm{\ispointsto{\mcode{y}}{239}{\imm}} }}
}
\end{equation*}
As the $\imm$ access permission prevents the synthesised code from
modifying the greyed heaplets, the synthesis search is forced to
back-track, picking an alternative substitution
$\subst' = [30/\mcode{z}]$ and converging on the desirable program
\mcode{*x = 30}.

\subsection{Composing Read-Only Borrows}
\label{sec:composing-borrows}

Having synthesised the \mcode{pick} function from
specification~\eqref{eq:pick}, we would like to use it in future
programs. For example, imagine that at some point, while synthesising
another program, we see the following as an intermediate goal:
\begin{equation}
  \label{eq:usepick}
\mcode{\set{u, v}; \teng
    {\asn{\ispointsto{u}{239}{\mut} \sep \ispointsto{v}{30}{\mut} \sep P}}
    {\asn{w \leq 200; \ispointsto{u}{w}{\mut} \sep
        \ispointsto{v}{w}{\mut} \sep Q}}
}
\end{equation}
It is clear that, modulo the names of the variables, we can synthesise
a part of the desired program by emitting a call \mcode{pick(u, v)},
which we can then reduce to the goal
$\mcode{\set{u, v} \teng {\asn{P}} {\asn{w \leq 200; Q}}}$
via an application of \framer.

Why is emitting such a call to \mcode{pick()} safe?
Intuitively, this can be done because the precondition of the
spec~\eqref{eq:pick} is \emph{weaker} than the one in the
goal~\eqref{eq:usepick}.
%
%
%
Indeed, the precondition of the latter provides the full (mutable)
access permission on the heap portion $\bl{\mcode{\ispointsto{v}{30}{\mut}}}$,
while the pre/postcondition of former requires a weaker form of
access, namely read-only: $\bl{\mcode{\ispointsto{y}{30}{\imm}}}$.
Therefore, our logical foundations should allow temporary
``downgrading'' of an access permission, \eg, from $\mut$ to $\imm$, for
the sake of synthesising calls.
While allowing this is straightforward and can be done similarly
to up-casting a type in languages like Java, what turns out to be less
trivial is making sure that the caller's initial stronger access
permission ($\mut$) is \emph{restored} once \mcode{pick(u, v)} returns.


\paragraph{Non-solutions.~}
Perhaps, the simplest way to allow the call to a function with a
weaker (in terms of access permissions) specification, would be to (a)
downgrade the caller's permissions on the corresponding heap fragments
to $\imm$, and (b) recover the permissions as per the callee's
specification. This approach significantly reduces the expressivity of
the logic (and, as a consequence, completeness of the synthesis). For
instance, adopting this strategy for using
specification~\eqref{eq:pick} in the goal~\eqref{eq:usepick} would
result in the unsolvable sub-goal of the form
$\mcode{\set{u, v}; \teng {\asn{\ispointsto{u}{30}{\mut} \sep
    \ispointsto{v}{30}{\imm} \sep P}} {\asn{\ispointsto{u}{30}{\mut}
    \sep \ispointsto{v}{30}{\mut} \sep Q}}}$.
This is due to the fact that the postcondition requires the
heaplet $\bl{\mcode{\ispointsto{v}{30}{\mut}}}$ to have the
write-permission~$\mut$, while the new precondition only provides the
$\imm$-access.


%
%
Another way to cater for a weaker callee's specification would be to
``chip out'' an $\imm$-permission from a caller's $\mut$-annotation
(in the spirit of fractional permissions), offer it to the callee,
and then ``merge'' it back to the caller's full-blown permission upon
return.
This solution works for simple examples, but not for heap predicates with
mixed permissions (discussion in \autoref{sec:related}).
%
%
Yet another approach would be to
create a ``read-only clone'' of the caller's $\mut$-annotation,
introducing an axiom of the form
$\bl{\mcode{\ispointsto{x}{t}{\mut} \bientails \ispointsto{x}{t}{\mut} \sep
  \ispointsto{x}{t}{\imm}}}$.
This way, the created component $\bl{\mcode{\ispointsto{x}{t}{\imm}}}$
could be provided to the callee and discarded
upon return since the caller retained the full permission of the original
heap.
Several works on read-only permissions have adopted this approach~\cite{Chargueraud-Pottier:ESOP17,Costea-al:PEPM14,David-Chin:OOPSLA11}.
While discarding such clones 
works just fine for sequential program verification, in the case of
synthesis which is guided by pre- and postconditions, \emph{incomplete}
postconditions could easily lead to intractable goals.

\paragraph{Our solution.~}
%
%
The key to gaining the necessary expressivity \wrt passing/returning
access permissions, while maintaining a sound yet simple logic, is
\emph{treating access permissions as first-class values}.
A natural consequence of this treatment is that immutability
annotations can be symbolic (\ie, variables of a special sort
``permission''), and the semantics of such variables is well
understood; we refer to these symbolic annotations as \emph{read-only
  borrows}.\footnote{In this regard, our symbolic borrows are very
  similar to abstract fractional permissions in \tname{Chalice} and
  \tname{VeriFast}~\cite{LeinoM09,Jacobs-al:NFM11}. We discuss the
  relation in detail in \autoref{sec:related}.}
%
%
%
For instance, using borrows, we can represent the
specification~\eqref{eq:pick} as an equivalent one:
\begin{equation}
  \label{eq:pic2}
\mcode{
  \mcode{\asn{ \ispointsto{x}{239}{\mut} \sep \ispointsto{y}{30}{\bimm{a}} }}
  ~\text{\code{void pick(loc x, loc y)}}~
  \mcode{\asn{z \leq 100; \ispointsto{x}{z}{\mut} \sep
      \ispointsto{y}{z}{\bimm{a}} }}
}
\end{equation}
The only substantial difference with spec~\eqref{eq:pick} is that now
the pointer \mcode{y}'s access permission is given an \emph{explicit name}
$\bl{\bimm{a}}$.
Such named annotations (\aka borrows) are treated as $\imm$ by the
callee, as long as the pure precondition does not constrain
them to be mutable.
%
%
However, giving these permissions names achieves an important goal: performing
accurate accounting while composing specifications with different
access permissions.
Specifically, we can now emit a call to \mcode{pick(u, v)} as specified
by~\eqref{eq:pic2} from the goal~\eqref{eq:usepick}, keeping in mind
the substitution
$\subst = [\mcode{u}/\mcode{x}, \mcode{v}/\mcode{y}, \mut/\bimm{a}]$. This call now
accounts for borrows as well, and makes it straightforward to restore
\mcode{v}'s original permission $\mut$ upon returning.

Following the same idea, borrows can be naturally composed through
capture-avoiding substitutions. For instance, the same
specification~\eqref{eq:pic2} of \mcode{pick} could be used to
advance the following modified version of the goal~\eqref{eq:usepick}:
\begin{equation*}
\mcode{
\set{u, v}; \teng
    {\asn{\ispointsto{\mcode{u}}{239}{\mut} \sep \ispointsto{v}{30}{\bimm{c}} \sep P}}
    {\asn{w \leq 210; \ispointsto{\mcode{u}}{w}{\mut} \sep
        \ispointsto{v}{w}{\bimm{c}} \sep Q}}
}
\end{equation*}
by means of taking the substitution
$\subst' = [\mcode{u}/\mcode{x}, \mcode{v}/\mcode{y}, \bimm{c}/\bimm{a}]$.


\subsection{Borrow-Polymorphic Inductive Predicates}
\label{sec:predicates}

Separation Logic owes its glory to the extensive use of
\emph{inductive heap predicates}---a compact way to capture the shape
and the properties of finite heap fragments corresponding to recursive
linked data structures.
Below we provide one of the most widely-used SL predicates, defining
the shape of a heap containing a null-terminated singly-linked list
with elements from a set~$\mcode{S}$:

{\footnotesize{
\begin{equation}
\label{eq:lseg}
\bl{
\begin{array}{r@{\ \ }c@{\ \ }r@{\ }c@{\ }l}
\mcode{\ls(x, S)} & \eqdef & \mcode{x = 0} & \wedge & \mcode{\set{ S \seteq \emptyset ; \emp}}
\\[3pt]
& | & \mcode{x \neq} 0 & \wedge & \mcode{\set{ S \seteq \set{\mcode{v}} \cup S_1 ;
\block{x}{2} \sep x \pts v \sep \angled{x, 1} \pts \nxt \sep \ls(\nxt, S_1) }}
\end{array}
}
\end{equation}
}}

The predicate contains two clauses describing the corresponding cases
of the list's shape depending on the value of the head pointer $\mcode{x}$. If
$\mcode{x}$ is zero, the list's heap representation is empty, and so is the
set of elements $\mcode{S}$.
Alternatively, if $\mcode{x}$ is not zero, it stores a record with two
items (indicated by the \emph{block assertion}
$\bl{\mcode{\block{x}{2}}}$), such that the \emph{payload} pointer
$\mcode{x}$ contains the value $\mcode{v}$ (where
$\mcode{S \seteq \set{\mcode{v}} \cup S_1}$ for some set
  $\mcode{S_1}$), and the pointer, corresponding to $\mcode{x + 1}$
  (denoted as $\bl{\mcode{\angled{x, 1}}}$) contains the address of
  the list's tail, $\nxt$.

While expressive enough to specify and enable synthesis of various
list-traversing and list-generating recursive functions via SSL, the
definition~\eqref{eq:lseg} does not allow one to restrict the access
permissions to different components of the list: all of the involved
memory locations can be mutated (which explains the synthesis issue we
described in \autoref{sec:weak}).
To remedy this weakness of the traditional SL-style predicates, we
propose to \emph{parameterise} them with read-only borrows, thus
making them aware of different access permissions to their various
components.
For instance, we propose to redefine the linked list predicate as
follows:
\vspace{-5pt}

{\footnotesize{
\begin{equation}
\label{eq:ls}
\bl{
\!\!\!\!\!\!\!\!
\begin{array}{l}
\mcode{\iipred{ls}{x, S}{\bimm{a}, \bimm{b}, \bimm{c}}} \eqdef
\\[2pt]
~~~~~~ \mcode{x = 0 \wedge  \set{ S \seteq \emptyset ; \emp}}
\\[2pt]
\mcode{
~~~~ | ~ x \neq 0 \wedge
\set{ S \seteq \set{\mcode{v}} \cup S_1 ;
\iblock{x}{2}{\bimm{a}} \sep \ispointsto{x}{v}{\bimm{b}} \sep
  \ispointsto{\angled{x, 1}}{\nxt}{\bimm{c}} \sep \iipred{ls}{\nxt,
  S_1}{\bimm{a}, \bimm{b}, \bimm{c}} }
}
\end{array}
}
\end{equation}
}}

The new definition~\eqref{eq:ls} is similar to the old
one~\eqref{eq:lseg}, but now, in addition to the standard predicate
parameters (\ie, the head pointer $\mcode{x}$ and the set $\mcode{S}$ in this
case), also features three borrow parameters $\bl{\bimm{a}}$,
$\bl{\bimm{b}}$, and $\bl{\bimm{c}}$ that stand as place-holders for
the access permissions to some particular components of the list.
Specifically, the symbolic borrows $\bl{\bimm{b}}$ and $\bl{\bimm{c}}$
control the permissions to manipulate the pointers $\mcode{x}$
and $\mcode{x + 1}$, correspondingly. The borrow $\bl{\bimm{a}}$,
modifying a block-type heaplet, determines whether the record starting
at $\mcode{x}$ can be deallocated with \mcode{free(x)}.
All the three borrows are passed in the same configuration to the
recursive instance of the predicate, thereby imposing the same
constraints on the rest of the corresponding list components.

Let us see the borrow-polymorphic inductive predicates in action.
Consider the following specification that asks for a function taking a
list of arbitrary values and replacing all of them with
zeroes:\footnote{We use $\mathbb{O}$ as a notation for a multi-set
  with an arbitrary finite number of zeroes.}
\begin{equation}
  \label{eq:reset}
  \mcode{\asn{ {\iipred{ls}{x,S}{\bimm{d}, \mut, \bimm{e}}}}}
  ~\text{\code{void reset(loc x)}}~
  \mcode{\asn{ {\iipred{ls}{x,\mathbb{O}}{\bimm{d},\mut, \bimm{e}}}}}
\end{equation}
The spec~\eqref{eq:reset} gives very little freedom to the function that
would satisfy it with regard to permissions to manipulate the contents
of the heap, constrained by the predicate
$\bl{\iipred{ls}{\mcode{x},S}{\bimm{d}, \mut, \bimm{e}}}$. As the first and
the third borrow parameters are instantiated with read-only borrows
($\bl{\bimm{d}}$ and $\bl{\bimm{e}}$), the desired function is not
going to be able to change the structural pointers or deallocate parts
of the list. The only allowed manipulation is, thus, changing the
values of the payload pointers.

This concise specification is pleasantly strong. To wit, in plain
\logic, a similar spec (without read-only annotations) would also
admit an implementation that fully deallocates the list or arbitrarily
changes its length. In order to avoid these outcomes, one would,
therefore, need to provide an alternative definition of the predicate
$\bl{\ls}$, which would incorporate the length property too.

Imagine now that one would like to use the implementation of
\mcode{reset} satisfying specification~\eqref{eq:reset} to generate a
function with the following spec, providing stronger access
permissions for the list components:
\begin{equation*}
  \mcode{\asn{ {\iipred{ls}{y,S}{\mut,\mut,\mut}}}}
  ~\text{\code{void call_reset(loc y)}}~
  \mcode{\asn{ {\iipred{ls}{y,\mathbb{O}}{\mut,\mut,\mut}}}}
\end{equation*}

%
%
During the synthesis of \mcode{call\_reset}, a call to \mcode{reset} is
generated. For this purpose the access permissions are borrowed and
recovered as per spec~\eqref{eq:reset} via the substitution
$[\mcode{y}/\mcode{x}, \mut/\mcode{d}, \mut/\mcode{e}]$ in a way
described in \autoref{sec:composing-borrows}.

\subsection{Putting It All Together}
\label{sec:all-together}

We conclude this overview by explaining how synthesis via \logic
enhanced with read-only borrows avoids the issue with spurious writes
outlined in~\autoref{sec:weak}.

To begin, we change the specification to the following one, which
makes use of the new list predicate~\eqref{eq:ls} and prevents any
modifications in the original list.
\begin{equation*}
  \label{eq:newcopy}
\mcode{
  \asn {\ispointsto{r}{x}{\mut} \sep {\iipred{ls}{x,S}{\bimm{a}, \bimm{b}, \bimm{c}}}}
  ~\text{\code{listcopy(r)}}~
  \asn{\ispointsto{r}{y}{\mut} \sep {\iipred{ls}{x,S}{\bimm{a},
        \bimm{b}, \bimm{c}}} \sep {\iipred{ls}{y,S}{\mut,\mut,\mut}}}
}
\end{equation*}
We should remark that, contrary to the solution sketched at the end of
\autoref{sec:weak}, which suggested using the predicate instance of
the shape $\bl{\mcode{\ls(x, S)[\imm]}}$, our concrete proposal does not allow
us to constrain the entire predicate with a single access permission
(\eg, $\imm$). Instead, we allow \emph{fine-grained} access
control to its particular elementary components by annotating each one
with an individual borrow.
The specification above allows the greatest flexibility \wrt
access permissions to the original list by giving them different names
($\bl{\bimm{a}}$, $\bl{\bimm{b}}$, $\bl{\bimm{c}}$).

In the process of synthesising the non-trivial branch of
\mcode{listcopy}, the search at some point will come up with the
following intermediate goal:
\begin{equation*}
{\small{
\begin{array}{l}
\set{\mcode{x}, \mcode{r}, \mcode{nxt}, \mcode{v}, \mcode{y12}};
\\
\\[-6pt]
\mcode{
{\asn{
S \seteq \set{v} \cup S_1;
\ispointsto{\mcode{r}}{\mcode{y12}}{\mut} \sep
\gbm{
\iblock{\mcode{x}}{2}{\bimm{a}} \sep
\ispointsto{\mcode{x}}{\mcode{v}}{\bimm{b}} \sep
\ispointsto{\angled{\mcode{x}, 1}}{\mcode{nxt}}{\bimm{c}}
} \sep
\iipred{ls}{\mcode{y12},S_1}{\mut, \mut, \mut} \sep
\ldots
}}
}
\\
\\[-6pt]
\leadsto~{\asn{
\gbm{
\mcode{
\iblock{{z}}{2}{\bimm{\mut}} \sep
\ispointsto{{z}}{\mcode{v}}{\mut} \sep
\ispointsto{\angled{{z}, 1}}{\mcode{y12}}{\mut}
}} \sep
\mcode{
\iipred{ls}{\mcode{y12},S_1}{\mut, \mut, \mut} \sep
\ldots
}
}}
\end{array}
}}
\end{equation*}
Since the logical variable $\mcode{z}$ in the postcondition is an
existential one, the greyed part of the symbolic heap can be satisfied
by either
(a) re-purposing the greyed part of the precondition
(which is what the implementation in \autoref{sec:weak} does), or
(b) allocating a corresponding record of two elements (as should be done).
With the read-only borrows in place, the unification of the
two greyed fragments in the pre- and postcondition via \hunir fails,
because the mutable annotation of
$\bl{\mcode{\ispointsto{{z}}{\mcode{v}}{\mut}}}$ in the post cannot be
matched by the read-only borrow
$\bl{\mcode{\ispointsto{\mcode{x}}{\mcode{v}}{\bimm{b}}}}$ in the
precondition.
Therefore, not being able to follow the derivation path (a), the
synthesiser is forced to explore an alternative one, eventually
deriving the version of \mcode{listcopy} without tail-swapping.

\section{\logicx: \logicfullx}
\label{sec:logic}

\begin{figure}[t]
\setlength{\belowcaptionskip}{-10pt}
\centering
\[
\!\!\!\!
{\small
\begin{array}{l@{\!\!}r@{\ \ }c@{\ \ }l}
  \text{Variable} & \multicolumn{2}{c}{x, y} & \text{Alpha-numeric identifiers}
  \\[2pt]
  \text{Size, offset}  & \multicolumn{2}{c}{\size, \off} &\text{Non-negative integers}
  \\[2pt]
  \text{Expression} & e & ::= &
  0 \mid \True \mid x \mid e = e \mid e \wedge e \mid \neg e
  \\[2pt]
  \text{Command} &
  c & ::= &
            \Letz~{x} = \deref{(x + \off)} \mid
            \deref{(x + \off)} = e \mid 
            \Letz~{x} = \malloc{\size} \mid
            \free{x} \\
            & & & \mid \cerror \mid f(\many{e_i})
            \mid c; c \mid 
        \Ifz (e) \set{c}~\Elsez \set{c}
  \\[2pt]
  \text{Fun. dict.} & \fctx & ::= & \epsilon \mid \fctx, f~(\many{x_i})~\set{~c~}
\end{array}
}
\]
\caption{Programming language grammar.}
\label{fig:lang}
\end{figure}

\begin{figure}[t]
  \centering
\[
\!\!\!\!\!\!\!\!\!\!\!\!
{\small
\begin{array}{l@{\ \ }l@{\ \ }c@{\ \ }l}
  \text{Pure logic term} & \phi, \psi, \chi, \alpha & ::= &
  0 \mid \True \mid \mut \mid \imm \mid x \mid \phi = \phi \mid \phi \wedge \phi \mid \neg \phi 
  \\[2pt]
  \text{Symbolic heap} &\mcode{ P, Q, R} & ::= &
  \mcode{\emp \mid \ipointsto{e, \off}{e}{\iiann}
  \mid
  \iblock{e}{\off}{\iiann} \mid \mcode{\ipred{p}{\many{\phi_i}}{\many{\iiann_i}}}
      \mid P \sep Q}
  \\[2pt]
  \text{Heap predicate} & \indp & ::= &
  \mcode{\ipred{p}{\many{x_i}}{\many{\ann_j}}~\many{\pclause{\sel_k}{\chi_k}{R_k}}}
  \\[2pt]
  \text{Function spec} & \fun & ::= &
  \mcode{ f(\many{x_i}) : \asn{\slP} \asn{\slQ}}
  \\[2pt]
  \text{Assertion} &  \slP, \slQ & ::= &
  \mcode{\set{\phi; P}}
  \\[2pt]
  \text{Environment} & \mcode{\env} & := &
  \mcode{  \epsilon \mid \env, x}
  \qquad ~
  \text{Context} \quad \mcode{\ctx} \quad  ~~~:= ~ \mcode{\epsilon \mid
    \ctx, \indp \mid \ctx, \fun}
  \\[2mm]
\end{array}
}
\]
\caption{\logicx assertion syntax.}
\label{fig:logic}
\end{figure}



We now give a formal presentation of \logicx---a version of \logic
extended with read-only borrows.
\autoref{fig:lang} and \autoref{fig:logic} present its programming and
assertion language, respectively.
For simplicity, we formalise a core language without theories (\eg,
natural numbers), similar to the one of
\tname{Smallfoot}~\cite{Berdine-al:APLAS05}; the only sorts in the
core language are locations, booleans, and permissions (where
permissions appear only in specifications) and the pure logic only has
equality.
In contrast, our implementation supports integers and sets
(where the latter also only appear in specifications),
with linear arithmetic and standard set operations.
We do not formalise sort-checking of formulae; however,
for readability, we will use the meta-variable \mcode{\alpha}
where the intended sort of the pure logic term is ``permission'',
and
\mcode{\iset} for the set of all permissions.
%
The permission to allocate or deallocate a memory-block
\mcode{\iblock{x}{n}{{\bimm{\iiann}}}}
is controlled by \mcode{\iiann}.
%

\subsection{\logicx rules}

New rules of \logicx are shown in \autoref{fig:irules}.
The figure contains only 3 rules: this
minimal adjustment is possible thanks to our approach to
unification and permission accounting from first principles.
Writing to a memory location requires its corresponding symbolic heap
to be annotated as mutable.
Note that for a precondition
\mcode{\asn{a = \mut ; \ipointsto{x}{5}{a}}},
a normalisation rule like \sleft would first transform it into  
\mcode{\asn{\mut = \mut ; \ipointsto{x}{5}{\mut}}},
at which point the \mcode{\writer} rule can be applied.
Note also that \mcode{\allocr} does not require specific permissions on the block in the postcondition;
if they turn out to be \mcode{\imm}, the resulting goal is unsolvable.


\begin{figure}[!t]
\setlength{\belowcaptionskip}{-10pt}
\centering
{\footnotesize
\begin{mathpar}
\inferrule[\writer]
%
%
{
\mcode{\vars{e} \subseteq \env}
\qquad
\mcode{e \neq e'}
\\
\mcode{\env; \trans{{\asn{\phi ;
             \ipointsto{{x,\off}}{e}{\mut} \sep P}}}
{\asn{\psi;  \ipointsto{{x,\off}}{e}{\mut} \sep Q}} {\prog}
}}
{
  \mcode{\env;
    \trans
    {\asn{\phi; \ipointsto{{x,\off}}{e'}{\mut}  \sep P}}
    {\asn{\psi;  \ipointsto{{x,\off}}{e}{\mut} \sep Q}}
    {\deref{(x + \off)} = e; \prog}
  }
}
\\
\inferrule[\allocr]
{
\mcode{  R {=} \iblock{z}{n}{{\iiann}} \sep \Sep_{0 \leq i < n}\paren{
    \ipointsto{z, i}{e_i}{\iiann_i}}
}
\qquad
\mcode{\paren{\set{y} \cup \set{\many{t_i}}} \cap \vars{\env, \slP, \slQ} {=} \emptyset}
\qquad
\mcode{z {\in} \ev{\env, \slP, \slQ}}
\\
\mcode{R' \eqdef \iblock{y}{n}{\mut} \sep \Sep_{0 \leq i < n}
  \paren{\ipointsto{y,i}{t_i}{\mut}}}
\\
\mcode{\ctx; \env; \trans{\asn{\phi; P \sep R'}}{\asn{\psi; Q \sep R}}{\prog}
}}
{
  \mcode{\ctx; \env;
  \trans
  {\asn{\phi; P}}
  {\asn{\psi; Q \sep \cached{R}}}
  {\Letz y = \malloc{n}; \prog}
}}
\\
\inferrule[\freer]
{
  \mcode{R = \iblock{x}{n}{{\mut}} \sep
    \Sep_{0 \leq i < n}\paren{
      \ipointsto{x,i}{e_i}{\mut} }}
 \qquad
 \mcode{\vars{\set{x} \cup \set{\many{e_i}}} \subseteq \env}
 \\
\mcode{\ctx; \env; \trans{\asn{\phi; P}}{\asn{\slQ}}{\prog}}
}
{\mcode{
\ctx; \env; \trans{\asn{\phi; P \sep \cached{R}}}{\asn{\slQ}}{\free{x}; \prog}
}}
\end{mathpar}
}
\caption{\logicx derivation rules.}
\label{fig:irules}
\end{figure}


Unsurprisingly, the rule for accessing a memory cell just for
reading purposes requires no adjustments since any
permission allows reading. Moreover, the  \mcode{\applyr} rule for
method invocation does not need adjustments either. Below, we
describe how borrow and return seamlessly operate within a method
call:
\begin{mathpar}
  \inferrule[\applyr]
  {
  \mcode{\fun \eqdef f(\many{x_i}) : \asn{\phi_f; P_f}\asn{\psi_f; Q_f} \in \ctx}
  \qquad
  \mcode{~R \ieqTags [\subst] P_f}
  \qquad
  \mcode{\valid{\phi \implies [\subst]\phi_f}}
  \quad
  \mcode{\many{e_i} = [\subst]\many{x_i}}
  \\
  \mcode{\vars{\many{e_i}} \subseteq \env}
  \quad
  \mcode{\phi' \eqdef [\subst]\psi_f}
  \qquad
  \mcode{ ~R' \eqdef [\subst]Q_f}
  \quad
  \mcode{\ctx; \env; \trans{\asn{\phi \wedge \phi'; P \sep R'}}{\asn{\slQ}}{\prog}}
}
  {
  \mcode{\ctx; \env; \trans{\asn{\phi; P \sep R}}{\asn{\slQ}}{f(\many{e_i}); \prog}}
}
\end{mathpar}

The $\applyr$ rule fires when a sub-heap \mcode{R} in the precondition of the
goal can be unified with the precondition \mcode{P_f} of a function
\mcode{f} from context \mcode{\ctx}.
Some salient points are worth mentioning here:
(1) the \emph{annotation borrowing} from \mcode{R} to \mcode{P_f} for
those symbolic sub-heaps in \mcode{P_f} which require read-only
permissions is handled by the unification of \mcode{P_f} with
\mcode{R}, namely \mcode{R \ieqTags \![\subst] P_f} (\ie, substitution
accounts for borrows: \mcode{\subs{\bimm{\ann}}{\iiann}});
(2) the \emph{annotation recovery} in the goal's new precondition
is implicit via
\mcode{R' \eqdef [\subst]Q_f}, where the substitution  \mcode{\subst}
was computed during the unification, that is, while borrowing;
(3) finding a substitution \mcode{\subst} for
\mcode{R \ieqTags \![\subst] P_f} fails if \mcode{R} does not
have sufficient accessibility permissions to call \code{f} (\ie,
substitutions of the form \mcode{\subs{\mut}{\bimm{\ann}}}
are disallowed since the domain of \mcode{\subst} may only contain
existentials).
We reiterate that read-only specifications only manipulate symbolic
borrows, that is to say, \mcode{\imm} constants are not expected in the specification.

\subsection{Memory Model}

We closely follow the standard SL memory
model~\cite{OHearn-al:CSL01,Reynolds:LICS02}
and assume \mcode{\text{Loc} \subset \text{Val}}.
 \begin{mathpar}
  \begin{array}{llcl}
    \text{(Heap)}
    & \mcode{h \in \text{Heaps}}
    & \mcode{::=}
    & \mcode{\text{Loc} \rightarrow \text{Val}}\\
    \text{(Stack)}
    & \mcode{s \in \text{Stacks}}
    & \mcode{::=}
    & \mcode{\text{Var} \rightarrow \text{Val}}
  \end{array}
\end{mathpar}

To enable C-like accounting of dynamically-allocated memory blocks,
we assume that the heap \mcode{h} can also store sizes of allocated blocks in dedicated locations.
Conceptually, this part of the heap corresponds to the meta-data of the memory allocator.
This accounting is required to ensure that only a previously allocated memory block can be disposed (as opposed to any set of allocated locations),
and also enables the \mcode{free} command to accept a single argument, the address of the block.
To model this meta-data, we introduce a function \mcode{bl\colon \text{Loc}\to \text{Loc}},
where \mcode{bl(x)} denotes the location in the heap where the block meta-data for the address \mcode{x} is stored,
if \mcode{x} is the starting address of a block.
In an actual language implementation, \mcode{bl(x)} might be, \eg, \mcode{x - 1}
(\ie, the meta-data is stored right before the block).

Since we have opted for an unsophisticated permission mechanism,
where the \emph{heap ownership is not divisible}, but some heap locations are restricted to
\mcode{\imm},
the definition of the satisfaction relation \mcode{\isatsl{}{}{\ilset}}
for the annotated assertions in a particular context \mcode{\ctx}
and given an interpretation \mcode{\interp}, is parameterised
with a fixed set of read-only locations, \mcode{\ilset}:

\begin{itemize}
\item \mcode{\isatsl{\angled{h, s}}{\asn{\phi; \emp}}{\ilset}}
  ~\Iff~
  \mcode{\denot{\phi}_s = \True} and \mcode{\dom{h} = \emptyset}.

\item \mcode{\isatsl{\angled{h, s}}{\asn{\phi;
        \ipointsto{e_1,\off}{e_2}{{\iiann}}}}{\ilset}}
  ~\Iff~
  \mcode{\denot{\phi}_s = \True}
  and \mcode{l\eqdef\denot{e_1}_s + \off}
  and \mcode{\dom{h} = \{l\}}
  and \mcode{h(l) = \denot{e_2}_s}
  %
  and \textcolor{black}{\mcode{l\in\ilset \Leftrightarrow \iiann=\imm}}
  .

\item \mcode{\isatsl{\angled{h, s}}{\asn{\phi; \iblock{e}{\size}{\iiann}}}{\ilset}}
  ~\Iff~
  \mcode{\denot{\phi}_s = \True}
  \textcolor{black}{
  and \mcode{l\eqdef bl(\denot{e}_s)}
  and \mcode{\dom{h} = \{l\}}
  and \mcode{h(l) = \size}
  and \mcode{l\in\ilset \Leftrightarrow \iiann=\imm}}.
  %

\item \mcode{\isatsl{\angled{h, s}}{\asn{\phi; P_1 \sep P_2 }}{\ilset}}  ~\Iff~
  \mcode{\exists~ h_1, h_2, h = h_1 \union h_2} and
  \mcode{\isatsl{\angled{h_1, s}}{\asn{\phi; P_1}}{\ilset}} and \\
  \mcode{\isatsl{\angled{h_2, s}}{\asn{\phi; P_2 }}{\ilset}}.

\item \mcode{\isatsl
    {\angled{h,s}}
    {\asn{\phi;\ipred{p}{\many{\psi_i}}{}}}
    {\ilset}}
  \Iff
  \mcode{\denot{\phi}_s = \True}
  and \mcode{\indp \eqdef \ipred{p}{\many{x_i}}{}
    \many{\pclause{\sel_k}{\chi_k}{R_k}} \in \ctx}
  and\\
  \mcode{\angled{h,\many{\denot{\psi_i}_s}} \in \interp(\indp)}
  and \mcode{
  {\bigvee_{k}(
    \isatsl{\angled{h, s}}{[\many{\subs{x_i}{\psi_i}}]\asn{\phi \wedge \sel_k \wedge \chi_k; R_k}}{\ilset}
    )}}.
\end{itemize}

There are two non-standard cases: points-to and block, whose
permissions must agree with \mcode{\ilset}. Note that in the
definition of satisfaction, we only need to consider that case where
the permission $\iiann$ is a value (\ie, either \mcode{\imm} or
\mcode{\mut}). Although in a specification $\iiann$ can also be a
variable, well-formedness guarantees that this variable must be
logical, and hence will be substituted away in the definition of
validity.

We stress the fact that a reference that has \mcode{\imm}
permissions to a certain symbolic heap still retains the full
ownership of that heap, with the restriction that it is not allowed to
update or deallocate it. Note that deallocation additionally requires
a mutable permission for the enclosing block.

\subsection{Soundness}

The \logicx operational semantics is in the spirit of the traditional
SL~\cite{Rowe-Brotherston:CPP17}, and hence is omitted for the sake of
saving space (selected rules are shown in
  App.~\ref{sec:os}).

%
The validity definition and the soundness proofs of \logic are ported
to \logicx without any modifications, since our current definition of
satisfaction implies the one defined for \logic:
\begin{definition}[Validity]
  \label{def:validity}
  We say that a well-formed Hoare-style specification
  \mcode{\ctx; \env; \asn{\slP}~\prog~\asn{\slQ}} is \emph{valid} \wrt the
  function dictionary \mcode{\fctx} \Iff
  whenever
  \mcode{\dom{s} = \env},
  \mcode{\forall \subst_{\text{gv}} = [\many{x_i \mapsto d_i}]_{x_i \in \gv{\env, \slP,
        \slQ}}}
  such that
  \mcode{\satsl{\angled{h, s}}{[\subst_{\text{gv}}]\slP}},  and
  \mcode{\opsemt{\fctx}{\angled{h, (c, s) \cdot \epsilon}}{\angled{h',
        (\ccd{skip}, s') \cdot \epsilon}}},
  it is also the case that
  \mcode{\satsl{\angled{h', s'}}{[\subst_{\text{ev}} \union
        \subst_{\text{gv}}]}\slQ}
  for some \mcode{\subst_{\text{ev}} = [\many{y_j \mapsto d_j}]_{y_j \in \ev{\env,
        \slP, \slQ}}}.
\end{definition}


The following theorem guarantees that, given a
program \mcode{\prog} generated with \logicx,
a heap model, and a set of read-only locations \mcode{\ilset}
that satisfy the program's precondition,
executing \mcode{\prog}
does not change those read-only locations:
\begin{theorem}[RO Heaps Do Not Change]
\label{thm:ro-heaps}
Given a Hoare-style specification
\mcode{\ctx; \env;
  {\asn{\phi;P}}
  {\prog}
  {\asn{\slQ}}
},
which is valid \wrt the function dictionary
\mcode{\fctx},
and a set of read-only memory locations \mcode{\ilset},
if:
  \begin{enumerate}[label=\emph{(\roman*)}]
  \item
    \mcode{\isatsl{\angled{h, s}}{[\subst]\slP}{\ilset}},
    for some
    \mcode{\mcode{h,s}}
    and
    \mcode{\subst},
    and
  \item
    \mcode{\opsemt{\fctx}
      {\angled{h, (c, s) \cdot \epsilon}}
      {\angled{h',(c', s') \cdot \epsilon}}}
    for some
    \mcode{h',s'}
    and
    \mcode{c'}
  \item
    \mcode{\ilset \subseteq \dom{h}}
  \end{enumerate}
  then
  \mcode{\ilset \subseteq \dom{h'}}
  and
  \mcode{\forall l \in \ilset,~ h(l) = h'(l).}
\end{theorem}

Starting from an abstract
state where a spatial heap has a
read-only permission, under no circumstance can this permission
be strengthened to \mcode{\mut}:
\begin{corollary}[No Permission Strengthening]
  \label{cr:strength}
  Given a valid Hoare-style specification
  \mcode{\ctx; \env; \asn{\phi; P}~c~\asn{\psi; Q}}
  and a permission
  \mcode{ {\iiann}},
  if
  \mcode{\psi \implies ({\iiann} = \mut)}
  then it is also the case that
  \mcode{\phi \implies ({\iiann} = \mut)}
.
\end{corollary}

As it turns out, permission weakening is possible, since,
though problematic, 
postcondition weakening is sound in general.
However, even though this affects completeness, it does
not affect our termination results. For example, given a synthesised auxiliary function
\mcode{\fun \eqdef f(x,r) :
  \asn{\ispointsto{x}{t}{\bimm{\ann_1}}
    \sep
    \ispointsto{r}{x}{\mut}
  }
  {\asn{\ispointsto{x}{t}{\bimm{\ann_2}}
    \sep
    \ispointsto{r}{t+1}{\mut}
  }}},
and a synthesis goal
\mcode{\ctx,\fun; \env;
  \trans
  {\asn{\ispointsto{x}{7}{\mut}
      \sep
      \ispointsto{y}{x}{\mut}
    }}
  {\asn{\ispointsto{x}{7}{\mut}
      \sep
      \ispointsto{y}{z}{\mut}
    }}
  {\prog}
},
firing the \mcode{\applyr} rule for the candidate function \mcode{f(x,r)}
would lead to the unsolvable goal
\mcode{\ctx,\fun; \env;}
\mcode{
  \trans
  {\asn{\ispointsto{x}{7}{\bimm{\ann_2'}}
      \sep
      \ispointsto{y}{8}{\mut}
    }}
  {\asn{\ispointsto{x}{7}{\mut}
      \sep
      \ispointsto{y}{z}{\mut}
    }}
  {f(x,y);\prog}
}.
\mcode{\framer} may never be fired on this new goal since the permission of
reference \mcode{x} in the goal's precondition has been permanently weakened.
To eliminate such sources of incompleteness we require the
user-provided predicates and specifications to be well-formed:
\begin{definition}[Well-Formedness of Spatial Predicates]
  \label{def:well-formed-pred}
  We say that a spatial predicate
  \mcode{\ipred{p}{\many{x_i}}{\many{\ann_j}}~\many{\pclause{\sel_k}{\chi_k}{R_k}}_{k\in 1..N}}
  is \emph{well-formed} \Iff \\
  \centerline{
    \mcode{(\bigcup_{k=1}^{N} (\vars{\sel_k} \cup \vars{\chi_k} \cup \vars{R_k} ) \cap \iset)
      \subseteq
      (\many{x_i} \cap \iset)}.
  }
\end{definition}
That is, every accessibility annotation within the predicate's clause is
bound by the predicate's parameters.

\begin{definition}[Well-Formedness of Specifications]
  \label{def:well-formed-spec}
  We say that a Hoare-style specification
  $\ctx; \env; \asn{\slP}~c~\asn{\slQ}$ is \emph{well-formed} \Iff
  \mcode{~\ev{\env,\slP, \slQ} \cap \iset = \emptyset}
  and every predicate instance
  in \mcode{\slP} and \mcode{\slQ} is an instance of a
  well-formed predicate.
\end{definition}
That is, postconditions are not allowed to have existential accessibility
annotations in order to avoid permanent weakening of accessibility.

A callee that requires borrows for a symbolic heap always returns back
to the caller its original permission for that respective symbolic
heap:
\begin{corollary}[Borrows Always Return]
  \label{cr:borrows}
  A heaplet with permission $\mcode{\iiann}$, either (a) retains the
  same permission \mcode{\iiann} after a call to a function that is
  decorated with well-formed specifications and that requires for that
  heaplet to have read-only permission, or (b) it may be deallocated
  in case if $\mcode{\iiann} = \mut$.
\end{corollary}



\section{Implementation and Evaluation}
\label{sec:evaluation}

We implemented \logicx in an enhanced version of the \tool tool, which
we refer to as \toolx~\cite{robosuslik}.\footnote{The sources are available at
  \url{https://github.com/TyGuS/robosuslik}.}
The changes to the original \tool infrastructure affected less than
100 lines of code.
The extended synthesis is backwards-compatible with the original \tool
benchmarks.
To make this possible, we treat the original \logic specifications
as annotated/instantiated with $\mut$ permissions, whenever necessary,
which is consistent with treatment of access permissions in \logicx.

We have conducted an extensive experimental evaluation of \toolx,
aiming to answer the following research questions:

\begin{enumerate}
\item Do borrowing annotations improve the \emph{performance} of
  SSL-based synthesis when using standard search
  strategy~\cite[\S~5.2]{Polikarpova-Sergey:POPL19}?

\item Do read-only borrows improve the \emph{quality} of synthesised programs,
  in terms of size and comprehensibility, \wrt to their counterparts
  obtained from regular, ``all-mutable'' specifications?

\item Do we obtain \emph{stronger correctness guarantees} for the
  programs from the standard SSL benchmark
  suite~\cite[\S~6.1]{Polikarpova-Sergey:POPL19} by simply adding,
  whenever reasonable, read-only annotations to their specifications?

\item Do borrowing specifications enable more \emph{robust} synthesis?
  That is, should we expect to obtain better programs/synthesis
  performance on average \emph{regardless} of the adopted unification
  and search strategies?


\end{enumerate}

\subsection{Experimental Setup}
\label{sec:experimental-setup}



\paragraph{Benchmark Suite.~}
To tackle the above research questions, we have adopted most of the
heap-manipulating benchmarks from \tool
suite~\cite[\S~6.1]{Polikarpova-Sergey:POPL19} (with some variations)
into our sets of experiments.
In particular we looked at the group of benchmarks which manipulate
singly linked list segments, sorted linked list segments and binary
trees.
We have manually modified the inductive definitions which describe
these data structures to use the more general concept of list
segments, as opposed to the original ones which solely employed the
definition of singly linked lists ending with \code{null} in their
evaluation setup.
We did not include the benchmarks concerning binary search trees (BSTs)
for the reasons outlined in the next paragraph.

\paragraph{The Tools.~}
For a fair comparison which accounts for the latest advancements to
\tool, we chose to parameterise the synthesis process with a flag that
turns the read-only annotations on and off (off means that they are set
to be mutable).
Those values which are the result of having this flag set will be
marked in the experiments with \mcode{\evalimm}, while those marked
with \mcode{\evalmut}
ignore the read-only annotations during the synthesis process. For
simplicity, we will refer to the two instances of the tool, namely
\mcode{\evalimm} and \mcode{\evalmut}, as two different tools. Each
tool was set to timeout after 2 minutes of attempting to synthesise a
program.

For the sake of this evaluation, we had to disable original \tool's
\emph{symmetry reduction}
optimisation~\cite[\S~5.3]{Polikarpova-Sergey:POPL19} for all the
experiments, as it turned out to be
{incomplete in the presence of new search strategies
(restoring completeness of this optimisation is orthogonal to this study,
and hence left to future work).}
%
Having this optimisation switched off resulted in slightly worse
absolute performance numbers for some of the ``vanilla'' benchmarks,
compared to that reported in the original work on \logic.
It also required us to exclude the BST benchmarks from the evaluation,
as the synthesis of those did exceed the timeout for both read-only
and mutable versions of the corresponding specifications.

\paragraph{Criteria.}
In an attempt to quantify our results,
we have looked at
the size of the synthesised program (\emph{AST size}),
the absolute time needed to synthesise the code given its
specification, averaged over several runs (\emph{Time}),
the number of backtrackings in the proof search due to nondeterminism
(\emph{\#Backtr}),
the total number of rule applications that the synthesis fired during
the search (\emph{\#Rules}), including those that lead to unsolvable
goals,
and the strength of the guarantees offered by the specifications
(\emph{Stronger Guarantees}).

\paragraph{Variables.}
Some benchmarks have shown improvement over the
synthesis process without the read-only annotations.
To emphasise the fact that read-only annotations' improvements are not
accidental, we have varied the inductive definitions of the corresponding
benchmarks to experiment with different
properties of the underlying structure:
the shape of the structure (in all the definitions),
the length of the structure (for those benchmarks tagged with \emph{len}),
the values stored within the structure (\emph{val}),
a combination
of all these properties (\emph{all})
as well as with the sortedness property for the ``Sorted list" group of
benchmarks.


\paragraph{Experiment Schema.}
To measure the performance and the quality of the borrowing-aware
synthesis we ran the benchmarks against the two different tools and
did a one-to-one comparison of the results.
We ran each tool three times for each benchmark, and average the
resulted values. Except for some small variations in the synthesis
time from one run to another, all other evaluation criteria remain
constant within all three runs.

To measure the tools' robustness we stressed the synthesis algorithm
by altering the default proof search strategy.
We prepared 42 such perturbations which we used to run against the
different program variants enumerated above. Each pair of program
variant and proof strategy perturbation has been then analysed to
measure the number of rules that had been fired by \mcode{\evalimm}
and \mcode{\evalmut}.

\paragraph{Hardware Setup.}
The experiments were conducted on a 64-bit machine running Ubuntu,
with an Intel Xeon CPU (6 cores, 2.40GHz) with 32GB RAM.

\subsection{Performance and Quality of the Borrowing-Aware Synthesis}
\label{sec:performance}

\begin{table}[t]
\begin{center}
\scriptsize

\begin{tabular}{@{} c | c |  c  c | c  c  c | c  c  c | c  c c | c  @{}}
\multirow{2}{*}{\head{Group}}
 &\multirow{2}{*}{\head{Description}}
 &\multicolumn{2}{c|}{ \head{AST size} } 
 & \multicolumn{3}{c|}{ \head{Time (sec)} } 
 & \multicolumn{3}{c|}{ \head{\#Backtr.} } 
 & \multicolumn{3}{c|}{ \head{\#Rules} } 
 &{\head{Stronger}}
 \\
 &
 &~\mcode{\evalimm}~&~\mcode{\evalmut}~
 &~\mcode{\evalimm}~&~\mcode{\evalmut}~& \mcode{\evalmut}{/}\mcode{\evalimm}
 &~\mcode{\evalimm}~&~\mcode{\evalmut}~& \mcode{\evalmut}{/}\mcode{\evalimm}
 &~\mcode{\evalimm}~&~\mcode{\evalmut}~& \mcode{\evalmut}{/}\mcode{\evalimm}
 &{\head{Guarant.}}
 \\ 
  \hhline{==============}  
  \multirow{9}{*}{\shortstack{Linked \\ List \\ Segment}}  &
                                                             \multirow{1}{*}{append}
  & 20 & 20 & \worse 1.5 & \worse  1.4 & \worse 0.9x & 8 & 8 & 1.0x & 77 & 78 & 1.0x& \stronger
 \\ 
  & \multirow{1}{*}{delete}  & 44 & 44 & \better 1.9 & \better 2.1 & \better 1.1x & 67 & 67 & 1.0x & 180 & 180 & 1.0x & \same
 \\ 
 & \multirow{1}{*}{dispose}  & 11 & 11 & 0.5 & 0.5 & 1.0x & 0 & 0 & 1.0x & 8 & 8 & 1.0x & \same
 \\ 
 & \multirow{1}{*}{init}  & 13 & 13 & 0.7 & 0.7 & 1.0x & 5 & 5 & 1.0x & 27 & 27 & 1.0x & \stronger
 \\ 
 & \multirow{1}{*}{lcopy}  & \better 32 & \better 35 & 1.0 & 1.0 & 1.0x & \better 9 & \better 14 & \better 1.5x& \better 66 & \better 82 & \better 1.2x& \stronger
 \\ 
 & \multirow{1}{*}{length}  & 22 & 22 & 1.5 & 1.5 & 1.0x & 2 & 2 & 1.0x & 38 & 38 & 1.0x & \stronger
 \\  
 & \multirow{1}{*}{max}  & 28 & 28 & \better 1.4 & \better 1.5 & \better 1.1x & 2 & 2 & 1.0x & 38 & 38 & 1.0x & \stronger
 \\ 
 & \multirow{1}{*}{min}  & 28 & 28 & 1.5 & 1.5 & 1.0x & 2 & 2 & 1.0x & 38 & 38 & 1.0x & \stronger
 \\ 
 & \multirow{1}{*}{singleton}  & 11 & 11 & 0.5 & 0.5 & 1.0x & 8 & 8 & 1.0x & 30 & 30 & 1.0x & \same
 \\ 
 
 \hline 
 \multirow{5}{*}{\shortstack{Sorted \\ List}}  & \multirow{1}{*}{ins-sort-all}  & 29 & 29 & 3.7 & 3.8 & 1.0x & 5 & 5 & 1.0x & 60 & 60 & 1.0x & \stronger
 \\ 
 & \multirow{1}{*}{ins-sort-len}  & 29 & 29 & 3.0 & 3.0 & 1.0x & \better 7 & \better 8 & \better 1.1x & \better 59 & \better 60 & \better 1.0x & \stronger
 \\ 
 & \multirow{1}{*}{ins-sort-val}  & 29 & 29 & 2.6 & 2.5 & 1.0x & 5 & 5 & 1.0x & 57 & 57 & 1.0x & \stronger
 \\ 
 & \multirow{1}{*}{insert}  & 53 & 53 & 7.8 & 8.0 & 1.0x & \better 35 & \better 96 & \better 2.7x & \better 214 & \better 338 & \better 1.6x & \stronger
 \\ 
 & \multirow{1}{*}{prepend}  & 11 & 11 & \better 0.5 & \better 0.6 & \better 1.2x & 1 & 1 & 1.0x & 17 & 17& 1.0x  & \stronger
 \\ 
 
 \hline 
 \multirow{14}{*}{Tree}  & \multirow{1}{*}{dispose}  & 16 & 16 & \better 0.4 & \better 0.5 & \better 1.2x & 0 & 0 & 1.0x & 10 & 10 & 1.0x & \same
 \\ 
 & \multirow{1}{*}{flatten-acc}  & 35 & 35 & 2.1 & 2.0 & 1.0x & 24 & 24 & 1.0x & 118 & 118 & 1.0x & \same
 \\ 
 & \multirow{1}{*}{flatten-app}  & 48 & 48 & 1.6 & 1.7 & 1.0x & 14 & 14 & 1.0x & 76 & 76 & 1.0x & \same
 \\ 
 & \multirow{1}{*}{morph}  & 19 & 19 & 0.6 & 0.5 &
                                                                 1.0x & 1 & 1 & 1.0x & 24 & 24 & 1.0x & \stronger
 \\ 
 & \multirow{1}{*}{tcopy-all}  & \better 42 & \better 51 & \better 1.5 & \better 2.2 & \better 1.5x & \better 10 & \better 88 & \better 8.8x & \better 85 & \better 296 & \better 3.5x & \stronger
 \\ 
 & \multirow{1}{*}{tcopy-len}  & \better 36 & \better 42 & \better 1.3 & \better 2.0 & \better 1.5x & \better 6 & \better 90 & \better 15x & \better 72 & \better 304 & \better 4.2x & \stronger
 \\ 
 & \multirow{1}{*}{tcopy-val}  & \better 42 & \better 51 & \better 1.4 & \better 5.3 & \better 3.8x & \better 10 & \better 1222 & \better 122x & \better 82 & \better 2673 & \better 32x & \stronger
 \\ 
 & \multirow{1}{*}{tcopy-ptr-all}  & \better 46 & \better 55 & \better 1.6 & \better 2.4  & \better 1.5x & \better 10 & \better 88 & \better 8.8x & \better 93 & \better 303 & \better 3.3x & \stronger
 \\  
 & \multirow{1}{*}{tcopy-ptr-len}  & \better 40 & \better 46 &\better  1.3 & \better 2.2 & \better 1.7x & \better 6 & \better 90 & \better 15x & \better 80 & \better 311 & \better 3.9x & \stronger
 \\ 
 & \multirow{1}{*}{tcopy-ptr-val}  & \better 46 & \better 55 & \better 1.3 & \better 5.8 & \better 4.5x & \better 10 & \better 1222 & \better 122x & \better 89 & \better 2679 & \better 30x & \stronger
 \\ 
 & \multirow{1}{*}{tsize-all}  & \better 32 & \better 38 & \worse 1.5 & \worse 1.4 &\worse 0.9x & \better 2 & \better 4 & \better 2.0x & \better 45 & \better 51 & \better 1.1x & \stronger
 \\ 
 & \multirow{1}{*}{tsize-len}  & 32 & 32 & \worse 1.2 & \worse 1.1 & \worse 0.9x & 2 & 2 & 1.0x & \better 44 & \better 46 & \better 1.0x & \stronger
 \\ 
 & \multirow{1}{*}{tsize-ptr-all}  & \better 36 & \better 42 & \worse 1.6 & \worse 1.4 &\worse 0.9x & \better 2 & \better 4 & \better 2.0x & \better 53 & \better 58 & \better 1.1x & \stronger
 \\ 
 & \multirow{1}{*}{tsize-ptr-len}  & 36 & 36 & 1.3 & 1.3 & 1.0x & 2 & 2 & 1.0x & \better 52 & \better 53 & \better 1.0x & \stronger\\ 
 \hhline{==============} 
  \end{tabular}

\end{center}
\caption{Benchmarks and comparison between the results for synthesis
  with read-only annotations (\mcode{\evalimm}) and without them
  (\mcode{\evalmut}). For each case study we measure the \head{AST
    size} of the synthesised program, the \head{Time} needed to
  synthesize the benchmark, the number of times that the synthesiser
  had to discard a derivation branch (\head{\#Backtr.}), and the total
  number of fired rules (\head{\#Rules}). 
%
}
\label{tab:performance}
\end{table}


\autoref{tab:performance} captures the results of running
\mcode{\evalimm} and \mcode{\evalmut} against the considered
benchmarks.
It provides the empirical proof that the borrowing-aware synthesis
improves the performance of the original SSL-based synthesis, or in
other words, answering positively the Research Question 1.

\mcode{\evalimm} suffers almost no loss in performance (except for a
few cases, such as the list segment \mcode{append} where there is a
negligible increase in time ), while the gain is considerable for those
synthesis problems with complex pointer manipulation. For example, if
we consider the number of fired rules as the performance measurement
criteria, in the worst case, \mcode{\evalimm} behaves the same as
\mcode{\evalmut}, while in the best scenario it buys us a 32-fold
decrease in the number of applied rules.
At the same time, synthesising a few small examples in the
\mcode{\evalimm} case is a bit slower, despite the same or smaller
number of rule applications. This is due to the increased number of
logical variables (because of added borrows) when discharging
obligations via SMT solver.
%

\autoref{fig:stats} offers a
statistical view of the numbers shown in the table, where smaller bars
indicate a better performance.
%
%
The barplots visually demonstrate that as the complexity of the
problem increases (approximately from left to right),
\mcode{\evalimm} outperforms \mcode{\evalmut} (
the bars
corresponding to \mcode{\evalmut} soar up compared to those of
\mcode{\evalimm}).

Perhaps the most important take-away from this experiment is that the
synthesis with read-only borrows often produces a more concise program
(light green cells in the columnt \emph{AST size} of
\autoref{tab:performance}), while retaining the same or better
performance \wrt all the evaluated criteria. For instance,
\mcode{\evalimm} gets rid of the spurious write from the motivating
example introduced in \autoref{sec:intro}, reducing the AST size from
35 nodes down to 32, while in the same time firing fewer rules. That
also means that we secure a positive answer for Research Question 2.

\begin{figure}[t!]
  \centering
  \vspace{-15pt}
  \includegraphics[width=1.0\textwidth]{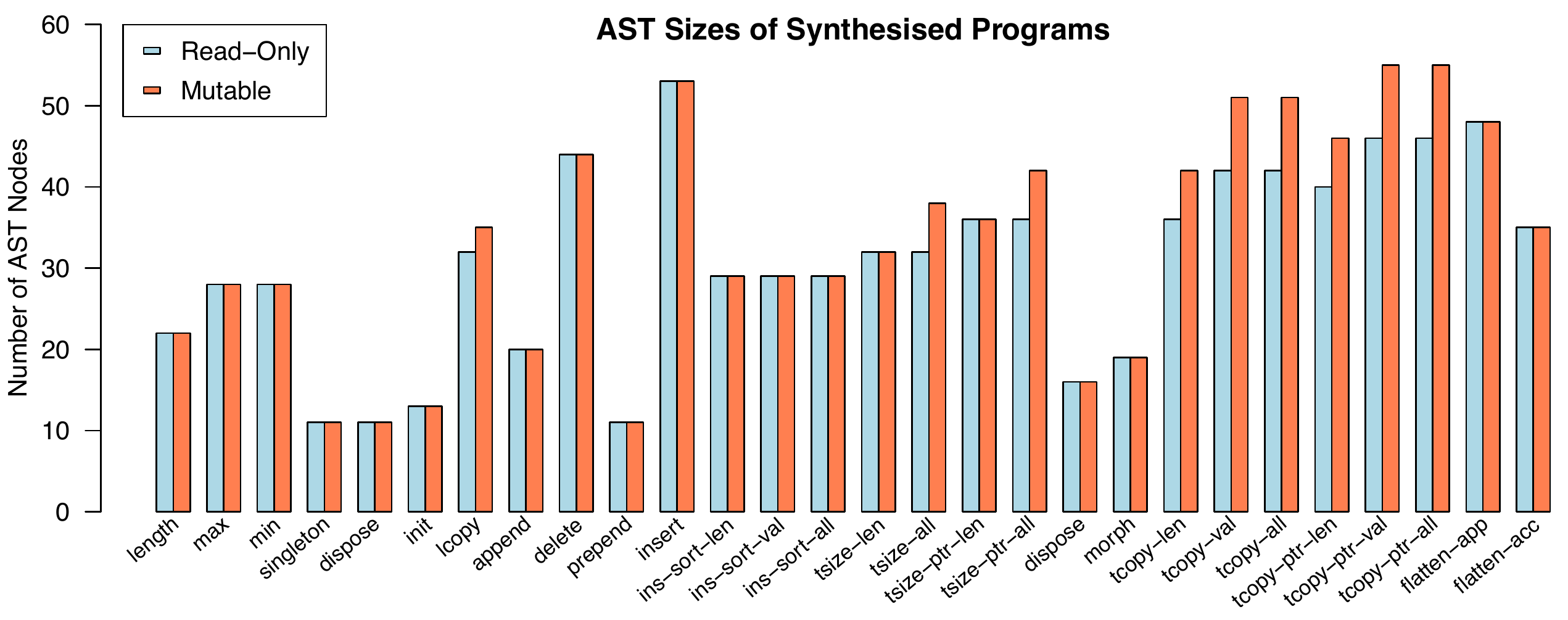}
  \includegraphics[width=1.0\textwidth]{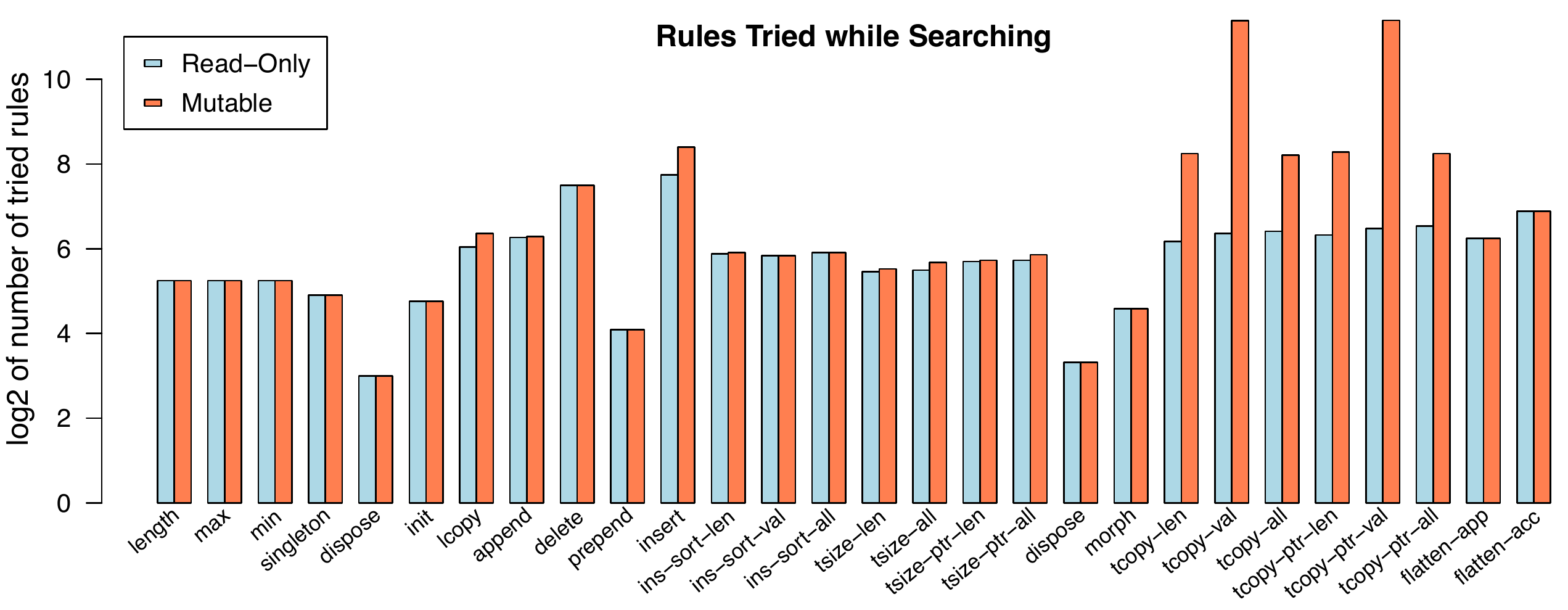}
  \caption{Statistics for synthesis with and without Read-Only specifications.}
  \label{fig:stats}
\end{figure}


\subsection{Stronger Correctness Guarantees}
\label{sec:strong-corr-guar}

To answer Research Question 3, we have manually compared the
correctness guarantees offered by the specifications that are
annotated with read-only permissions against the default ones - the results
are summarized in the last column of \autoref{tab:performance}.
For instance, a specification stating that the shape of a linked-list
segment is read-only implies that the size of that segment remains
constant through the program's execution. In other words, there is no
need for the length property to be captured separately in the
segment's definition. If, in addition to the shape, the payload of the
segment is also read-only, then the set of values and their ordering
are also invariant.

Consider the goal {\mcode{ \teng
    {\asn{\iipred{lseg}{x,y,s}{\bimm{\ann_1},\bimm{\ann_2},\bimm{\ann_3}}}}
    {\asn{\iipred{lseg}{x,y,s}{\bimm{\ann_1},\bimm{\ann_2},\bimm{\ann_3}}}}}},
where \mcode{\lseg} is a generalised inductive definition of a list
segment which ends at \mcode{y} and contains the set of values
\mcode{s}.
The borrowing-aware synthesiser will produce a program which is
guaranteed to treat the segment pointed by \mcode{x} and ending with
\mcode{y} as read-only (that is, its shape, length, values and
orderings are invariant).
At the same time, for a goal {\mcode{ \teng
    {\asn{\spred{lseg}{x,y,s}}} {\asn{\spred{lseg}{x,y,s}}}}} assuming
all list components are mutable, the guarantees are that the
returned segment still ends in \mcode{y} and contains values
\mcode{s}.
Such a goal, however, does not disallow internal modifications of the
segment, such as reordering and duplicating list elements.

The few entries marked with \mcode{\same} are programs with
specifications which have not got stronger when instrumented with
\mcode{\evalimm} annotations (\eg, \code{delete}). These benchmarks
require mutation over the entire data structure, hence the read-only
annotations do not influence the offered guarantees.
Overall, our observations that read-only annotations offer stronger
guarantees are in agreement with the works on SL-based program
verification~\cite{David-Chin:OOPSLA11,Chargueraud-Pottier:ESOP17},
but are promoted here to the more challenging problem of program
synthesis.

\subsection{Robustness under Synthesis Perturbations}
\label{sec:robust-unification}

There is no single search heuristics that will work equally well for
any given specification: for a particular fixed search strategy, a
synthesiser can exhibit suboptimal performance for some goals, while
converging quickly on some others. By evaluating robustness \wrt to
$\imm$ and $\mut$ specification methodologies, we are hoping to show
that, provided a large variety of ``reasonable'' search heuristics,
read-only annotations deliver better synthesis performance ``on
average''.

For this set of experiments, we have focused on four characteristic
programs from our performance benchmarks based on their pointer
manipulation complexity: list segment copy (\code{lcopy}), insertion
into a sorted list segment (\code{insert}), copying a tree
(\code{tcopy}), and a variation of the tree copy that shares the same
pointer for the input tree and its returned copy (\code{tcopy{-}ptr}).

\paragraph{Exploring Different Unification Orders.~}
Since spatial unification stays at the core of the synthesis
process, we implemented 6 different strategies for choosing a
unification candidate based on the following criteria: the size of the
heaplet chunk (we experimented with both favoring the smallest heap as
well as the largest one as the best unification candidate), the name
of the predicate (we considered both an ascending as well as a
descending priority queue), and a customised ranking function which
associates a cost to a symbolic heap based on its kind---a block is
cheaper to unify than a points-to which in turn is cheaper than a
spatial predicate.

\paragraph{Exploring Different Search Strategies.~}
We next designed 6 strategies for prioritising the rule applications.
One of the crux rules in this matter, is the $\writer$ rule
whose different priority schemes might make all the
results seem randomly-generated.
In
the cases where $\writer$ leads to unsolvable goals, one might
rightfully argue that \mcode{\evalimm} has a clear advantage over
\mcode{\evalmut} (\emph{fail fast}). However, for the cases where
mutation leads to a solution faster, then \mcode{\evalmut} might have
an advantage over \mcode{\evalimm} (\emph{solve fast}). Because these
are just intuitive observations, and for fairness sake, we
experimented with both the cases where $\writer$ has a high and a low
priority in the queue of rule
phases~\cite[\S~5.2]{Polikarpova-Sergey:POPL19}.
Since most of the benchmarks involve recursion, we also chose to
shuffle around the priorities of the $\invkind$ and $\applyr$ rules.
Again, we chose between a stack high and a bottom low priority for
these rules to give a fair chance to both tools.


We considered all combinations of the 6 unification permutations and
the 6 rule-application permutations (plus the default one) to obtain
42 different proof search perturbations. We will use the following
notation in the narrative below:
\begin{itemize}
  \item \mcode{S} is the set comprising the synthesis
    problems: \code{lcopy}, \code{insert}, \code{tcopy},
    \code{tcopy-ptr}.
  \item \mcode{V} is the set of all specification variations: \emph{len},
    \emph{val}, \emph{all}.
  \item \mcode{K} is the set of all 42 possible tool perturbations.
\end{itemize}

\begin{figure}[t!]
  \centering
  \includegraphics[width=0.98\textwidth]{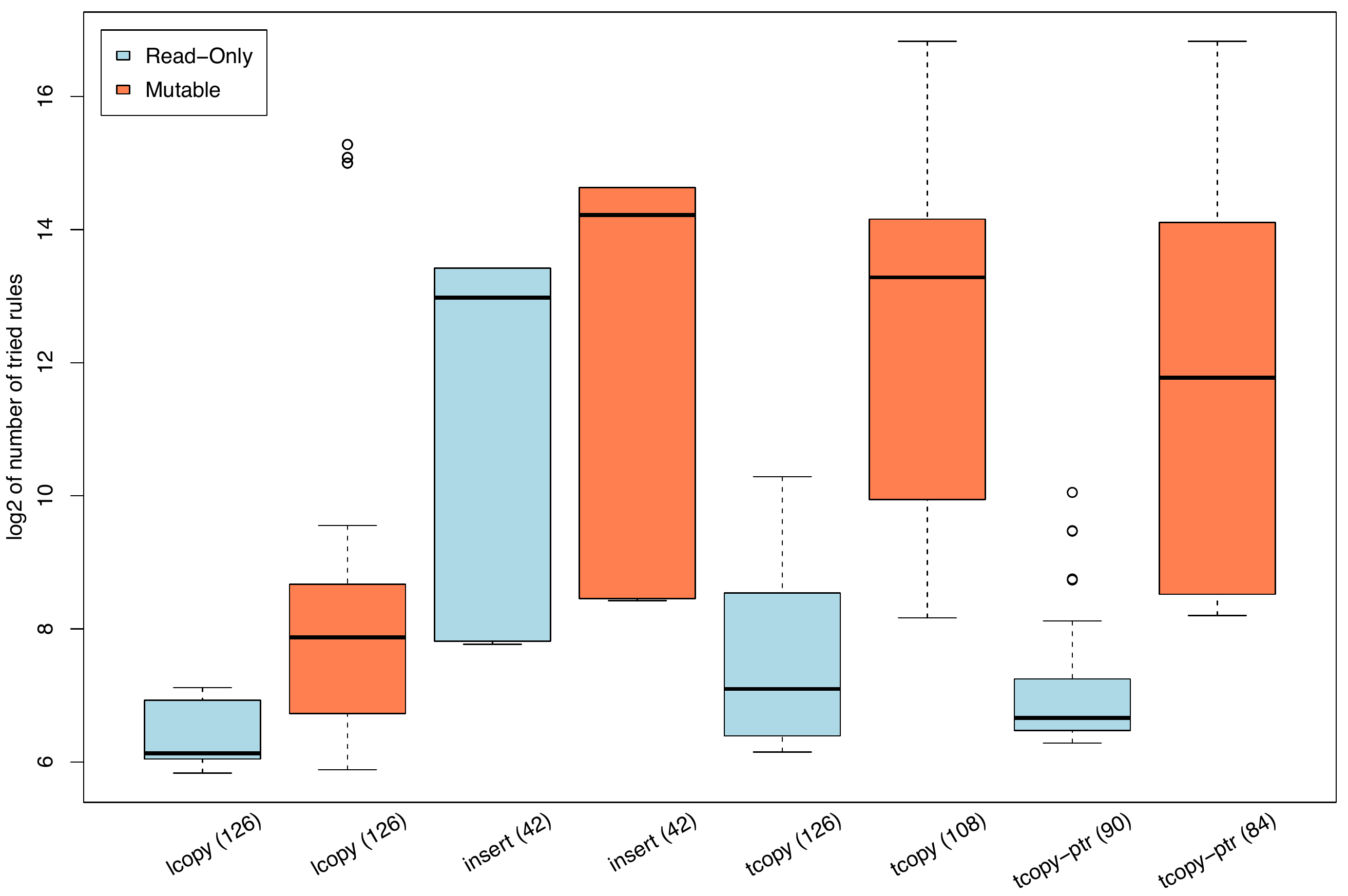}
  \caption{Boxplots of variations in
    $\log_2(\text{numbers of applied rules})$ for synthesis
    perturbations. Numbers of data points for each example are given
    in parentheses.}
  \label{fig:robust}
\end{figure}

The distributions of the number of rules fired for
each tool (\mcode{\evalimm} and \mcode{\evalmut})
with the 42 perturbations
over the 4 synthesis problems with 3 variants of specification each,
that is 1008 different synthesis runs,
are summarised using the boxplots in \autoref{fig:robust}.
%
%
There is a boxplot corresponding to each pair of tool and synthesis
problem. In the ideal case, each boxplot contains 126 data points
corresponding to a unique combination \mcode{(v,k)} of a specification
variation \mcode{v \in V} and a tool perturbation \mcode{k \in K}. A
boxplot is the distribution of such data based on a six number
summary: minimum, first quartile, median, third quartile, maximum,
outliers. For example, the boxplot for \code{tcopy-ptr} corresponding
to \mcode{\evalimm} and containing 90 data points, reads as follows:
``the synthesis processes fired between 64 and 256 rules, with most of
the processes firing between 64 and 128 rules. There are three
exception where the synthesiser fired more than 256 rules". Note that
the y-axis represents the binary logarithm of the number of fired
rules.

Even though we attempted to synthesise each program 126 times for each
tool, some attempts hit the timeout and therefore their corresponding
data points had to be eliminated from the boxplot.
It is of note, though, that whenever \mcode{\evalimm} with
configuration \mcode{(v,k)} hit the timeout for the synthesis problem
\mcode{s \in S}, so did \mcode{\evalmut}, hence both the
\mcode{(\evalimm,s,(v,k))} as well as \mcode{(\evalmut,s,(v,k))} are
omitted from the boxplots.
But the inverse did not hold: \mcode{\evalimm} hit the timeout fewer
times than \mcode{\evalmut}, hence \mcode{\evalimm} is measured at
disadvantage (\ie, more data points means more opportunities to show
worse results).
Since \code{insert} collected the highest number of
timeouts, we equalised it to remove non-matched entries across the two
tools.

\begin{figure}[t!]
  \centering
  \setlength{\belowcaptionskip}{-10pt}
  \setlength{\lineskip}{-0pt}
  \includegraphics[width=1.0\textwidth]{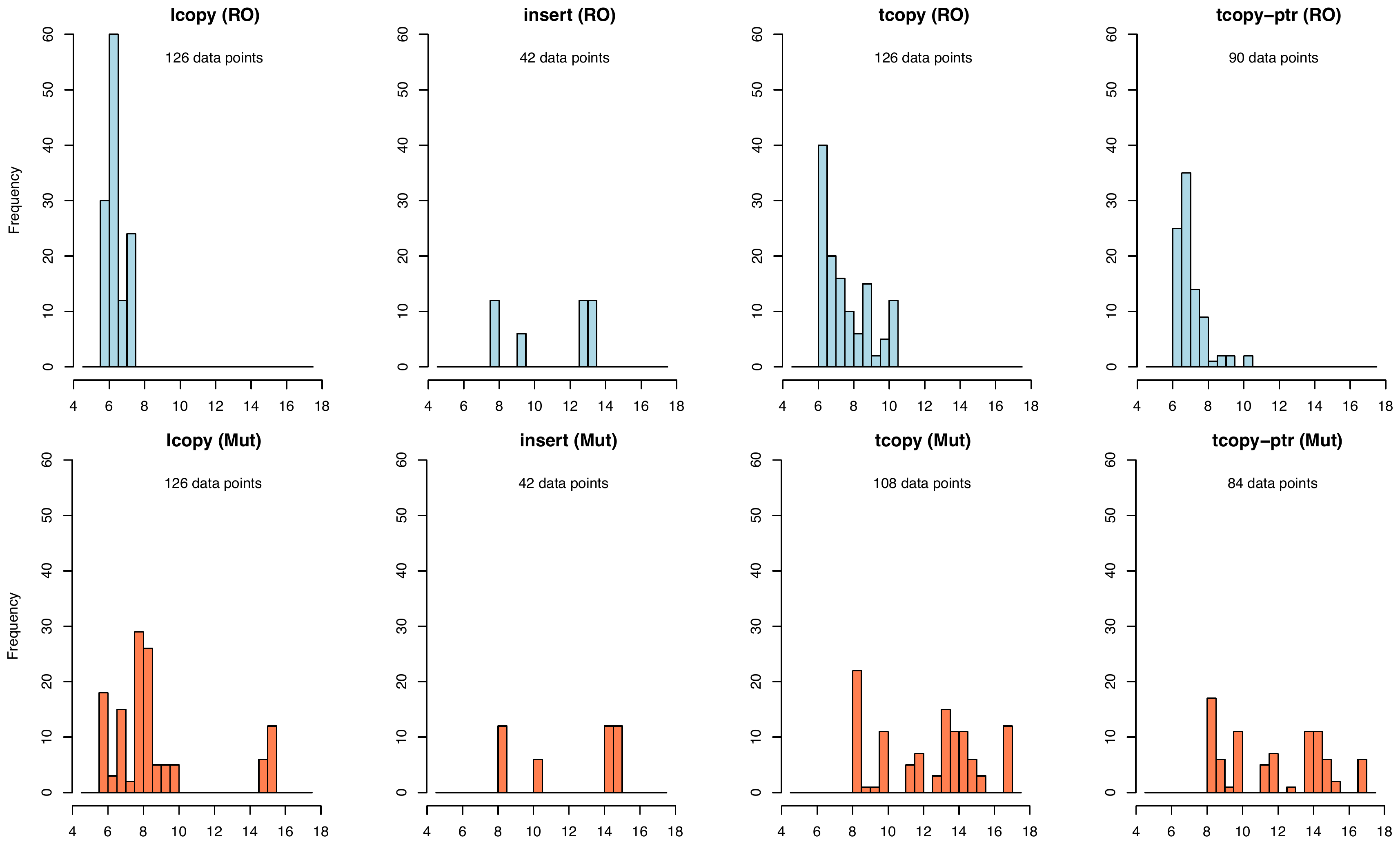}
  \caption{Distributions of $\log_2(\text{number of attempted rule applications})$.}
  \label{fig:hist}
\end{figure}

Despite \mcode{\evalimm}'s potential measurement disadvantage, the
boxplots depicts it as a clear winner. Not only \mcode{\evalimm} fires
fewer rules in all the cases, but with the exception of \code{insert},
it is also more stable to the proof search perturbations, it varies a
few order of magnitude less than \mcode{\evalmut} does for the same
configurations. \autoref{fig:hist} supports this observation by
offering a more detailed view on the distributions of the numbers of
fired rules per synthesis configuration.
%
%
Taller bars show that more processes fall in the same range (\wrt the
number of fired rules). For \code{lcopy}, \code{tcopy},
\code{tcopy-ptr} it is clear that \mcode{\evalmut} has a wider
distribution of the number of fired rules, that is, \mcode{\evalmut}
is \emph{more sensitive} to the perturbations than \mcode{\evalimm}.
We additionally make some further observations:

\begin{itemize}
\item Despite a similar distribution \wrt the numbers of fired rules
  in the case of \code{insert}, \mcode{\evalimm} produces compact ASTs
  of size 53 for all perturbations, while \mcode{\evalmut} fluctuates
  between producing ASTs of size 53 and 62.

\item For all the synthesis tasks, \mcode{\evalimm} produced the same
  AST irrespective of the tool's perturbation. In contrast, there were
  synthesis problems for which \mcode{\evalmut} produced as many as 3
  different ASTs for different perturbations, none of which were as
  concise as the one produced by \mcode{\evalimm} for the same
  configuration. A manual investigation of the synthesised programs
  revealed that all the extra AST nodes produced by \mcode{\evalmut}
  are the consequence of spurious write statements.

\item The outliers of (\mcode{\evalmut},~\code{lcopy}) are
  ridiculously high, firing close to 40k rules.

\item The outliers of (\mcode{\evalimm},~\code{tcopy}) are still below the
  median values of (\mcode{\evalmut},~\code{tcopy}).

\item Except for \code{insert}, the best performance of
  \mcode{\evalmut}, in terms of fired rules, barely overlaps with the
  worst performance of \mcode{\evalimm}.

\item Except for \code{insert}, the medians of \mcode{\evalimm}
  are 
  closer to the lowest value of the data distribution,
  as opposed to \mcode{\evalmut}
  where the tendancy is to fire more rules.

\item In absolute values, \mcode{\evalimm} hit the 2-minutes timeout
  102 times compared to \mcode{\evalmut}, which hit the timeout 132
  times.

\end{itemize}

\noindent
We believe that the main take-aways from this set of experiments,
along with the positive answer to the Research Question 4, are as
follows:
\begin{itemize}

\item \mcode{\evalimm} is more stable \wrt the number of rules fired
  and the size of the generated AST for many reasonable proof search
  perturbations.

\item \mcode{\evalimm} produces better programs, which avoid spurious
  statements, irrespective of the perturbation and number of rules
  fired during the search.
\end{itemize}












\section{Limitations and Discussion}
\label{sec:limitations}


\paragraph{Flexible aliasing.~}
Separating conjunction asserts that the heap can be split into
two disjoint parts, or in other words it carries an implicit
non-aliasing information.
Specifically,
\mcode{\spointsto{x}{\_} \sep \spointsto{y}{\_}} states that \mcode{x} and
\mcode{y} are non-aliased. Such assertions can be used to specify methods
as below:
\[
  {{
      \mcode{\asn{\spointsto{x}{n} \sep \spointsto{y}{m} \sep \spointsto{ret}{x}}}
      ~\text{\code{sum(x, y, ret)}}~
      \mcode{\asn{\spointsto{x}{n} \sep \spointsto{y}{m} \sep \spointsto{ret}{n+m}}}
    }}
  \]
Occasionally, enforcing \mcode{x} and \mcode{y} to be non-aliased is
too restrictive, rejecting safe calls such as \mcode{sum(p,p,q)}.
%
Approaches to support immutable annotations
permit such calls without compromising 
safety if both pointers, aliased or not, are annotated as read-only~\cite{David-Chin:OOPSLA11,Chargueraud-Pottier:ESOP17}.
\logicx does not support such flexible aliasing.

\paragraph{Precondition strengthening.~}
Let us assume that
\mcode{\iipred{srtl}{x,n,lo,hi}{{\iiann_1},{\iiann_2},{\iiann_3}}} is
an inductive predicate that describes a sorted linked list of size \mcode{n} with
\mcode{lo} and \mcode{hi} being the list's minimum and maximum payload
value, respectively.
Now, consider the following synthesis goal:
\vspace{-3pt}
\[
{{\mcode{\set{x, y}
      \teng
    {\asn{y \pts {x} \sep \iipred{srtl}{x,n,lo,hi}{\mut,\mut,\mut}}}
    {\asn{y \pts {n} \sep \iipred{srtl}{x,n,lo,hi}{\mut,\mut,\mut}}}}}.}
\]
As stated, the goal clearly requires the program to compute the length
\mcode{n} of the list. Imagine that we already have a function that
does precisely that, even though it is stated in terms of a list
predicate that does not enforce sortedness:
\vspace{-3pt}
\[
{{
\mcode{\asn{ret\pts{x} \sep \iipred{ls}{x,n}{\bimm{\ann_1},\bimm{\ann_2},\bimm{\ann_3}} }}
~\text{\code{length(x, ret)}}~
\mcode{\asn{ret\pts{n} \sep
    \iipred{ls}{x,n}{\bimm{\ann_1},\bimm{\ann_2},\bimm{\ann_3}}}}
}}
\]
%
%
To solve the initial goal, the synthesiser could weaken the
given precondition \mcode{\iipred{srtl}{x,n,lo,hi}{\mut,\mut,\mut}} to
\mcode{\iipred{ls}{x,n}{\mut,\mut,\mut}}, and then successfully
synthesise a call to the \code{length} method.
Unfortunately, the resulting goal, obtained after having emitted the
call to \code{length} and applying $\framer$, is unsolvable:
\vspace{-3pt}
\[
  {{\mcode{\set{x, y}
        \teng
        {\asn{ \iipred{ls}{x,n}{\mut,\mut,\mut}}}
        {\asn{ \iipred{srtl}{x,n,lo,hi}{\mut,\mut,\mut}}}}}.}
\]
%
%
\noindent since the logic does not allow to strengthen an arbitrary
linked list to a \emph{sorted} linked list without retaining the prior
knowledge.
%
%
Should we have adopted an alternative treatment of RO
permissions~\cite{David-Chin:OOPSLA11,Chargueraud-Pottier:ESOP17}
allowing the caller to retain the full permission of the sorted list,
the postcondition of \code{length} would \emph{not} contain the
list-related heap component and would only quantify over the result
pointer: 
{
\[
  {{
      \mcode{\asn{ret\pts{x} \sep \imm{(\iipred{ls}{x,n}{})} }}
      ~\text{\code{length(x, ret)}}~
      \mcode{\asn{ret\pts{n}
          }}
    }}
  \]
}
thus leading to the solvable goal below:
\[
  {\mcode{\set{x, y};
        \teng
        {\asn{\iipred{srtl}{x,n,lo,hi}{\mut,\mut,\mut}}}
        {\asn{\iipred{srtl}{x,n,lo,hi}{\mut,\mut,\mut}}}}}.
\]
One straightforward way for \logicx to cope with this limitation is to
simply add a version of \code{length} annotated with
specifications that cater to \mcode{srtl}.


\paragraph{Overcoming the limitations.~}
While the ``caller keeps the permission'' kind of approach would buy us
flexible aliasing and calls with weaker specifications, it would compromise
the benefits discussed earlier with respect to the granularity of
borrow-polymorphic inductive predicates.
One possible solution to gain the best of both worlds
would be to design a permission system which allows both borrow-polymorphic
inductive predicates as well as read-only modalities to co-exist, where the
latter would overwrite the predicate's mixed permissions.
In other words, the read-only modality enforces a read-only treatment of the
predicate irrespective of its permission arguments, while the permission
arguments control the treatment of a mutable predicate.
The theoretical implications of such a design choice are
left as part of future work.

\paragraph{Extending read-only specifications to concurrency.~}
Thus far we have only investigated the synthesis of sequential
programs, for which read-only annotations helped to reduce the synthesis cost.
%
%
%
%
Assuming that the synthesiser has the capability to synthesise
concurrent programs as well, the borrows annotation mechanism in its current form may not be
able to cope with general resource sharing.
This is because a
callee which requires read-only permissions to a
particular symbolic heap still 
consumes the entire required symbolic heap from the caller, despite the read-only
requirement; 
hence, there is no space left for sharing.
That said, the recently proposed alternative approaches to introduce
read-only
annotations~\cite{David-Chin:OOPSLA11,Chargueraud-Pottier:ESOP17} have
no formal support for heap sharing in the presence of concurrency
either.
To address these challenges, we could  adopt a more
sophisticated approach based on fractional permissions
mechanism~\cite{Bornat-al:POPL05,HeuleLMS13:VMCAI13,Muller-al:VMCAI16,Boyland:SAS03,Le-Hobor:ESOP18},
but this is left as part of future work since it is orthogonal to the
current scope.

\section{Related Work}
\label{sec:related}

\paragraph{Language design.}
There is a large body of work on integrating access permissions into
practical type
systems~\cite{Garcia:2014,Balabonski-al:TOPLAS16,Stork:2014} (see,
\eg, the survey by Clarke~\etal~\cite{Clarke2013}).
One notable such system, which is the closest in its spirit to our
proposal, is the borrows type system of the Rust programming
language~\cite{rust-borrows} proved safe with \tname{RustBelt}~\cite{jung2017rustbelt}.
Similar to our approach, borrows in Rust are short-lived: in Rust they
share the scope with the owner; in our approach they do not escape the
scope of a method call.
In contrast with our work, Rust's type system carefully manages
different references to data by imposing strict sharing constraints,
whereas in our approach the treatment of aliasing is taken care of
automatically by building on Separation Logic.
Moreover, Rust allows read-only borrows to be duplicated, while in the
sequential setting of \logicx this is currently not possible.

Somewhat related to our approach, Naden~\etal propose a mechanisms
for borrowing permissions, albeit integrated as a fundamental part of
a type system \cite{NadenBAB12:POPL12}.
%
Their type system comes equipped with {\it change permissions}
which enforce the borrowing requirements and describe the effects of the borrowing upon return.
As a result of treating permissions as first-class values,
we do not need to explicitly describe the flow of permissions for each borrow
since this is controlled by a mix of the substitution and unification principles.


%


\paragraph{Program verification with read-only permissions.}
%
Boyland introduced fractional permissions to statically reason about
interference in the presence of  {\it shared-memory concurrency}~\cite{Boyland:SAS03}.
A permission \mcode{p} denotes full 
resource ownership (i.e. read-write access) when \mcode{p=1}, while
\mcode{p \in (0,1)} denotes a partial 
ownership (i.e. read-only access).
To leverage permissions in practice, a system must support two key operations:
permission splitting and permission borrowing.
Permission splitting (and merging back) follows the split rule:
\mcode{\ispointsto{x}{a}{p} =  \ispointsto{x}{a}{p_1} \sep
  \ispointsto{x}{a}{p_2}}, with \mcode{p =
  p_1 + p_2} and \mcode{p,p_1,p_2 \in (0,1]}.
Permission borrowing refers to the safe manipulation of permissions:
a callee may remove some permissions from the caller, use them temporarily,
and give them back upon return.

Though it exists, tool support for fractional permissions is still
scarce. Leino and M\"{u}ller introduced a mechanism for storing fractional
permissions in data structures via dedicated access predicates in the
\tname{Chalice} verification tool~\cite{LeinoM09}.
To promote generic specifications, Heule \etal advanced \tname{Chalice}
with instatiable abstract permissions,
allowing automatic fire of the split rule and symbolic borrowing~\cite{HeuleLMS13:VMCAI13}.
\tname{VeriFast}~\cite{Jacobs-al:NFM11} is guided by contracts written in
Separation Logic and assumes the existence of lemmas to
cater for permission splitting.
\tname{Viper}~\cite{Muller-al:VMCAI16} is an intermediate language
which supports various permission models, including
abstract fractional permissions~
\cite{Astrauskas0PS19,Summers018}.
%
Similar to \tname{Chalice}, the permissions are attached to memory
locations using an accessibility predicate. 
To reason about it, 
\tname{Viper} uses permission-aware assertions and assumptions,
which correspond in our
approach to the unification and the substitution operations, respectively.
Like \tname{Viper}, we enhance the basic memory constructors, that is blocks and
points-to, to account for permissions, but in contrast, the
\mcode{\applyr} rule in our approach is standard, \ie, \emph{not}
permission-aware.
%
%

These tools, along with others~\cite{GordonPPBD12:OOPSLA12,Appel:ESOP11}, offer strong
correctness guarantees in the presence of resource sharing.
However, there is a class of problems, namely those involving predicates
with mixed permissions, whose guarantees are weakened
due to the general fractional permissions model behind these tools.
%
We next exemplify this class of problems in a sequential setting.
We start by considering a method which resets the values stored
in a linked-list while maintaining its shape
(\mcode{p<1} below is to enforce the immutable shape):
\[
{
  \mcode{
    \asn{p<1 ; ~\ls(x, S)[1,p]}
    ~\text{\code{void reset(loc x)}}~
    \asn{\ls(x, \{0\})[1,p]}
  }.
}
\]
Assume a call to this method, namely \mcode{reset(y)}.
The caller has full permission over the
entire list passed as argument, that is \mcode{\ls(y, B)[1,1]}.
This attempt leads to two issues.
The first has to do with splitting the payload's permission
(before the call) such that it matches the callee's postcondition.
To be able to modify the list's payload, the callee must get the payload's full ownership,
hence the caller should retain 0:
\mcode{\ls(y, B)[1,1] = \ls(y, B)[0,1/2] * \ls(y, B)[1,1/2]}.
But \mcode{0}  is not a valid fractional permission.
The second issue surfaces while attempting to merge the permissions after the call:
\mcode{\ls(y, B)[0,1/2] * \ls(y, \{0\})[1,1/2]} is invalid
since the two instances of \mcode{ls} have incompatible arguments (namely \mcode{B} and \mcode{\{0\}}).
To avoid such problems, \logicx abandons the split rule
and instead always manipulates full ownership of resources, hence it does not
use fractions.
This compromise, along with the support for symbolic borrows,
allows \toolx to guarantee read-only-ness in a sequential setting
while avoiding the aforementioned issues.
%
More investigations are needed
in order to lift this result to concurrency reasoning.
Another feature which distinguishes the current work from those based on
fractional permissions, is the support for permissions as parameters of the
predicate, which in turn supports the definition of predicates with mixed permissions.
%

Immutable specifications on top of Separation Logic
have also been studied
by David and Chin~\cite{David-Chin:OOPSLA11}.
%
Unlike our approach which treats borrows as polymorphic
variables that rely on the basic concept of substitution, their
annotation mechanism comprises only constants and requires a specially
tailored entailment on top of enhanced proof rules.
Since callers retain the heap ownership upon calling a method with
read-only requirements, their  machinery supports flexible aliasing
and cut-point preservation---features that we could not find a good use for in the
context of program synthesis.
%
%
An attempt to extend David and Chin's work by adding
support for predicates with mixed permissions~\cite{Costea-al:PEPM14}
suffers from significant annotation overhead. Specifically, it employs
a mix of mutable, immutable, and \emph{absent} permissions, so that
each mutable heaplet in the precondition requires a corresponding
matching heaplet annotated with absent in the postcondition.

Chargu{\'{e}}raud and Pottier~\cite{Chargueraud-Pottier:ESOP17}
extended Separation Logic with (non-fractional) read-only assertions
that can be freely duplicated or discarded.
The intuition for that is that the heap ownership is not
transferred when calling a method which requires read-only access.
Their approach creates lexically-scoped copies of the
RO-permissions before emitting a call, which, in turn, involves discarding the
corresponding heap from the postcondition to guarantee a sound read-only
modality.
%
%
Adapting this modality to program synthesis guided by pre- and postconditions,
would require a completely new system of deductive synthesis since
most of the rules in \logic are not designed to handle the discardable read-only heaps.
In contrast, \logicx supports permission-parametric predicates
(\eg,~\eqref{eq:ls}) requiring only minimal adjustments to its host
logic, \ie, \logic.
%


\paragraph{Program synthesis.}
\logicx continues a long line of work on program synthesis from formal specifications~\cite{Srivastava-al:POPL10,Leino-Milicevic:OOPSLA2012,Solar-Lezama:STTT13,Torlak-Bodik:PLDI14,Qiu-SolarLezama:OOPSLA17}
and in particular, \emph{deductive synthesis}~\cite{Manna-Waldinger:TOPLAS80,Kneuss-al:OOPSLA13,Delaware-al:POPL15,Polikarpova-al:PLDI16,Polikarpova-Sergey:POPL19},
which can be characterised as search in the space of \emph{proofs} of program correctness
(rather than in the space of programs).
Most directly \logicx builds upon our prior work on
\logic~\cite{Polikarpova-Sergey:POPL19} and enhances its specification
language with read-only annotations.
In that sense, the present work is also related to various approaches that use \emph{non-functional} specifications as input to synthesis.
It is common to use \emph{syntactic} non-functional specifications, such as grammars~\cite{Alur-al:FMCAD13}, sketches~\cite{Solar-Lezama:STTT13,Qiu-SolarLezama:OOPSLA17}, or restrictions on the number of times a component can be used~\cite{Gulwani-al:PLDI11}.
More recent work has explored \emph{semantic} non-functional specifications, including type annotations for resource consumption~\cite{KnothWPH19}
and security/privacy~\cite{Lifty16,GasconTCM17,Smith2019}.
This research direction is promising because
(a) annotations often enable the programmer to express a strong specification concisely, and
(b) checking annotations is often more compositional (\ie, fails faster) than checking functional specifications, which makes synthesis more efficient.
In the present work we have demonstrated that both of these benefits
of non-functional specifications also hold for the read-only
annotations of \logicx.



\section{Conclusion}
\label{sec:conclusion}

In this work, we have advanced the state of the art in program
synthesis by highlighting the benefits of guiding the synthesis
process with information about memory access permissions.
We have designed the logic \logicx and implemented the tool \toolx,
showing that a minimalistic discipline for read-only permissions
already brings significant improvements \wrt the performance and
robustness of the synthesiser, as well as \wrt the quality of its
generated programs.

\paragraph{Acknowledgements.~}

We thank Alexander J. Summers, Cristina David, Olivier Danvy, and
Peter O'Hearn for their comments on the prelimiary versions of the
paper.
We are very grateful to the ESOP 2020 reviewers for their detailed
feedback, which helped to conduct a more adequate comparison with
related approaches and, thus, better frame the conceptual
contributions of this work.

Nadia Polikarpova's research was supported by NSF grant 1911149.
Amy Zhu's research internship and stay in Singapore during the Summer
2019 was supported by Ilya Sergey's start-up grant at Yale-NUS
College, and made possible thanks to UBC Science Co-op Program.

\bibliographystyle{plain}
\bibliography{references}

\appendix
\section{Operational Semantics}
\label{sec:os}

The operational semantics of \logicx is given below.

{\center{
    \begin{tabular}{c}

\vspace{15pt}

{\centering{
\begin{tabular}{c}
\\
$  \ebpsem{\writer}{{
      \Hypo{\mcode{
          (\denot{x}_s + \off) {\in} \dom{h}
          \qquad
          h'\eqdef h[(\denot{x}_s + \off) \rightarrow \denot{e}_s]
        }}
      \Infer1{\mcode{
          \opsem
          {\fctx}
          {\angled{h,(\deref{(x + \off)} = e ;\prog,s) \cdot S}}
          {\angled{h',(\prog,s) \cdot S}}
        }}
    }}
$  
\\\\
$
  \ebpsem{\allocr}{{
      \Hypo{\mcode{
          y {\notin} \dom{s}
          \quad
          \exists l . bl(l) \notin \dom{h}
          \quad
          \forall ~ i {\in} 0 {\ldots} \size - 1,
          ~
          (l{+}i) {\notin} \dom{h}         
        }}
      \Infer[no rule]1{\mcode{
          s' \eqdef s[y \rightarrow l]
          \quad
          h' \eqdef h\union (bl(l)\rightarrow \size) \union [(l{+}i) \rightarrow \_]
        }}
      \Infer1{\mcode{
          \opsem
          {\fctx}
          {\angled{h,(\Letz y = \malloc{\size};\prog,s) \cdot S}}
          {\angled{h',(\prog,s') \cdot S}}
        }}
    }}
$
\\\\ 
$
  \ebpsem{\freer}{{
      \Hypo{\mcode{
          l \triangleq \denot{x}_s \quad
          \size \triangleq h(bl(l))
        }}
      \Infer[no rule]1{\mcode{
          h = h'\union (bl(l)\rightarrow \size) \union [(l{+}i) \rightarrow \_]
        }}
      \Infer1{\mcode{
          \opsem
          {\fctx}
          {\angled{h,(\free{x};\prog,s) \cdot S}}
          {\angled{h',(\prog,s) \cdot S}}
        }}
    }}
$
  \\ \\
$
  \ebpsem{\applyr}{{
      \Hypo{\mcode{
          f(\many{x_i})\{\prog'\} \in \fctx
          \quad
          s' \eqdef [x_i \rightarrow \denot{y_i}_s]
        }}
      \Infer1{\mcode{
          \opsem
          {\fctx}
          {\angled{h,(f(\many{y_i});\prog,s) \cdot S}}
          {\angled{h', (\prog',s') \cdot (\prog,s) \cdot S}}
        }}
    }}
$
\end{tabular}
}}

      
    \end{tabular}
}}

The semantics is fairly
traditional, except for the block accounting, reflected in the rules
\mcode{\allocr} and \mcode{\freer}).
A machine built over this model takes no decision based on the borrows
information, hence such a machine accepts more programs than \logicx.
Moreover, in comparison to \logic, \logicx conservatively narrows the
space of accepted derivations.
Consequently, the termination results of the former directly apply to
the later.

\newpage

\section{Soundness}
\textbf{Theorem~\ref{thm:ro-heaps} (RO Heaps Do Not Change).}
\emph{
Given a Hoare-style specification
\mcode{\ctx; \env;
  {\asn{\phi;P}}
  {\prog}
  {\asn{\slQ}}
},
which is valid \wrt the function dictionary
\mcode{\fctx},
and a set of read-only memory locations \mcode{\ilset},
if:
  \begin{enumerate}[label=\emph{(\roman*)}]
  \item
    \mcode{\isatsl{\angled{h, s}}{[\subst]\slP}{\ilset}},
    for some 
    \mcode{\mcode{h,s}}
    and
    \mcode{\subst},
    and
  \item
    \mcode{\opsemt{\fctx}
      {\angled{h, (c, s) \cdot \epsilon}}
      {\angled{h',(c', s') \cdot \epsilon}}}
    for some 
    \mcode{h',s'}
    and
    \mcode{c'}
  \item
    \mcode{\ilset \subseteq \dom{h}}
  \end{enumerate}
  then 
  \mcode{\ilset \subseteq \dom{h'}}
  and
  \mcode{\forall l \in \ilset,~ h(l) = h'(l).}
}
\begin{proof}
  Our proof strategy is to prove that if
  \mcode{\isatsl{}{}{\ilset}} holds and
  \mcode{\Opsem} modifies some memory locations, then those
  locations are not in \mcode{\ilset}.
  The theorem then follows by induction on \mcode{\Opsemt}.
  For brevity, we only show the base cases for the heap update:\\
  \noindent{\bf Case:} \mcode{\writer}. Assume that
  (1) \mcode{\isatsl{\angled{h, s}}
    {[\subst]\ipointsto{x + \off}{\_}{\iiann} \sep P'}{\ilset}},
  and also that 
  (2) \mcode{\opsemt{\fctx} 
    {\angled{h, (\deref{(x + \off)} = e;\prog, s) \cdot \epsilon}}
    {\angled{h',(\prog, s') \cdot \epsilon}}}.
  Assuming in (1) that \mcode{\iiann \neq \mut},
  contradicts the validity assumption of
  \mcode{\ctx; \env;
    {\asn{\phi;P}}
    {\deref{(x + \off)} = e;\prog}
    {\asn{\slQ}}
  }, therefore \mcode{\iiann=\mut}.
  By (i) and \mcode{\iiann=\mut}, it follows that 
  (3) \mcode{(\denot{x}_s + \off) \notin \ilset}.
  By (1) it follows that 
  (4) \mcode{h'=h[(\denot{x}_s + \off) \rightarrow \denot{e}_s]}.
  Finally, by (3) and (4) it follows that
  \mcode{\forall l \in \ilset,~ h(l) = h'(l).}
  In addition, since \mcode{dom(h) = dom(h')}, then from (iii) it follows that \mcode{\ilset \subseteq \dom{h'}}.\\
  \noindent{\bf Case:} \mcode{\allocr}. Assuming that
  \mcode{\opsemt{\fctx} 
    {\angled{h, (\Letz y = \malloc{\size};\prog, s) \cdot \epsilon}}
    {\angled{h',(\prog, s') \cdot \epsilon}}},
  it follows then that 
  (1)
  {\mcode{
      \forall ~ i {\in} 0 {\ldots} \size{-}1,
      ~
      (l{+}i) {\notin} \dom{h}
      ~\text{and}~
      bl(l) \notin \dom(h)
      ~\text{and}~
      h' \eqdef h\union(bl(l) \rightarrow \size)\union[(l{+}i) \rightarrow \_]
    }}.
  From (1) and (iii) we conclude that
  that \mcode{\forall l \in \ilset,~ h(l) = h'(l).}
  In addition, since \mcode{dom(h) \subset dom(h')}, then from (iii) it follows that \mcode{\ilset \subseteq \dom{h'}}.\\
  \noindent{\bf Case:} \mcode{\freer}. 
  Assuming that
  \mcode{
          \opsem
          {\fctx}
          {\angled{h,(\free{x};\prog,s) \cdot S}}
          {\angled{h',(\prog,s) \cdot S}}
        }
  we get (1) \mcode{h = h'\union (bl(l)\rightarrow \size) \union [(l{+}0) \rightarrow \_ \ldots (l{+}\size - 1) \rightarrow \_]} 
  (where \mcode{l \triangleq \denot{x}_s} and \mcode{\size \triangleq h(bl(l))}):
  hence the values stored in the heap are not modified, and it is only left to prove that \mcode{\ilset \subseteq \dom{h'}},
  \ie that \mcode{bl(l) \notin \ilset} and \mcode{(l{+}i) \notin \ilset} for \mcode{i \in 0\ldots \size - 1}.
  By validity and (i), we have (2) 
  \mcode{\isatsl{\angled{h, s}}{[\subst](\iblock{x}{n}{{\mut}} \sep
    \Sep_{0 \leq i < n}\paren{
      \ipointsto{x,i}{e_i}{\mut} } ) \sep P'}{\ilset}}.
  The rest follows from (2) and the definition of satisfaction.
  \qed
\end{proof}
  
\noindent
\textbf{Corollary ~\ref{cr:strength} (No Permission Strenghtening).}
\emph{	Given a valid Hoare-style specification
	\mcode{\ctx; \env; \asn{\phi; P}~c~\asn{\psi; Q}}
	and a permission
	\mcode{ {\iiann}},
    if
   \mcode{\psi \implies ({\iiann} = \mut)}
   then it is also the case that 
   \mcode{\phi \implies ({\iiann} = \mut)}
}
\begin{proof} The proof is a direct consequence of the soundness
  theorem which disallows post-condition strengthening.
  %
\qed
\end{proof}


\end{document}